\documentclass[11pt, letterpaper]{amsart}
\usepackage{color}
\usepackage{enumerate}
\usepackage{amssymb}
\usepackage{epsfig,psfrag,latexsym}
\usepackage{multirow}
\usepackage{tabularx}
\usepackage{appendix}
\usepackage{caption}
\usepackage[normalem]{ulem}
\usepackage[top=1.0in,bottom=1.0in,right=1.00in,left=1.00in]{geometry}

\newtheorem{thm}{Theorem}[section]

\newtheorem{lem}[thm]{Lemma}
\newtheorem{prop}[thm]{Proposition}

\theoremstyle{definition}

\newtheorem{ass}[thm]{Assumption}

\theoremstyle{remark}

\newtheorem{rem}[thm]{Remark}

\numberwithin{equation}{section}

\newcommand{\esp}{\mathbb{E}}

\newcommand{\reals}{\mathbb R}

\newcommand{\B}{\mathcal{B}}

\newcommand{\F}{\mathcal{F}}

\newcommand{\Lcal}{\mathcal{L}}

\newcommand{\wh}{\widehat}

\newcommand{\eps}{\varepsilon}

\newcommand{\qtau}{\mathfrak{u}}

\newcommand{\such}{\ | \ }
\newcommand{\dfn}{\, := \,}

\newcommand{\delmf}{\frac{\delta}{m}}

\newcommand{\ito}{\textrm{It\^o}}

\newcommand{\tbabm}{b^{\textrm{abm}}}

\newcommand{\nada}[1]{}

\newcommand{\bra}[1]{\left[#1\right]}
\newcommand{\cbra}[1]{\left\{#1\right\}}
\newcommand{\dbra}[1]{[\kern-0.15em[ #1 ]\kern-0.15em]}
\newcommand{\dbraco}[1]{[\kern-0.15em[ #1 [\kern-0.15em[}
\newcommand{\wt}[1]{\widetilde{#1}}

\newcommand{\ol}[1]{\overline{#1}}
\newcommand{\ul}[1]{\underline{#1}}
\newcommand{\rcpdiff}[2]{\frac{1}{#1} - \frac{1}{#2}}

\newcommand{\rcpdiffb}[4]{\frac{#1}{#2}-\frac{#3}{#4}}

\newcommand{\prob}{\mathbb{P}}
\newcommand{\qprob}{\mathbb{Q}}

\newcommand{\basisq}{\left(\Omega, \mathcal{F}, \mathbb{F}, \qprob\right)}

\newcommand{\filt}{\mathbb{F}}
\newcommand{\espalt}[3]{\mathbb{E}^{#1}_{#2}\left[#3\right]}

\newcommand{\condexpv}[4]{\esp^{#1}_{#2}\bra{#3\such #4}}
\newcommand{\condespalt}[3]{\esp^{#1}_{#3}\bra{#2}}

\title[Mortgage Contracts]{Mortgage Contracts and Underwater Default}

\author{Yerkin Kitapbayev}
\address{Department of Mathematics\\
North Carolina State University\\
Raleigh, NC 27607}
\email{ykitapb@ncsu.edu}

\author{Scott Robertson}
\address{Questrom School of Business\\
Boston University\\
Boston, MA 02215}
\email{scottrob@bu.edu}
\thanks{S. Robertson is supported in part by the National Science Foundation under grant number DMS-1613159.}

\date{\today}

\begin{document}

\begin{abstract}

We analyze recently proposed mortgage contracts that aim to eliminate selective borrower default when the loan balance exceeds the house price (the ``underwater'' effect).  We show contracts that automatically reduce the outstanding balance in the event of house price decline remove the default incentive, but may induce prepayment in low price states.  However, low state prepayments vanish if the benefit from home ownership is sufficiently high. We also show that capital gain sharing features, such as prepayment penalties in high house price states, are ineffective as they virtually eliminate prepayment. For observed foreclosure costs, we find that contracts with automatic balance adjustments become preferable to the traditional fixed-rate contracts at mortgage rate spreads between 20-50 basis points. We obtain these results for perpetual versions of the contracts using American options pricing methodology, in a continuous-time model with diffusive home prices. The contracts' values and optimal decision rules are associated with free boundary problems, which admit semi-explicit solutions.
\end{abstract}

\maketitle

\section{\bf Introduction}\label{S:intro}

It is by now incontrovertible that the housing crisis of 2007-2009 was exacerbated by the ``underwater'' effect, where homeowners owed more on their house than it was worth on the market.  The negative effects of being underwater are well known, having been documented at the government (\cite{Fed_Subprime}), academic (\cite{andersson2014loss}) and public (\cite{mian2015house}) levels.

Underwater mortgages powered a vicious cycle within many United States metropolitan areas, most prominently in the Southwest. Borrowers, having purchased homes initially worth far more than their incomes could support, but recently having lost a large portion of their value, were stuck in houses which they could neither afford nor sell.  In response, they engaged in large scale selective defaults on their loans (c.f. \cite{bhutta2010depth}).  This led banks to incur significant losses, either directly through the foreclosure process, or indirectly through the resultant fire sales, in which the repossessed home was sold at a depressed value  (c.f. \cite{hatcher2006foreclosure, campbell2011forced, andersson2014loss}). The fire sales further depressed home prices and appraisal values, putting more homeowners under water, repeating the cycle.

In short, underwater mortgages posed, and continue to pose, significant risks for the homeowner, the lending institution, and the broader health of the economy.  Furthermore, there is an asymmetry in that traditional mortgage contracts have built-in protections against interest rate movements (e.g. adjustable rate mortgages, refinancing with no penalties), but there are no such protections for house price decline.  Indeed, default associated to house price decline has traditionally been considered a ``moral'' issue (\cite{1996Ctam}), to be worked out in the (lengthy, expensive) legal system.

To mitigate both risks and costs associated to underwater mortgages,  a number of alternative mortgage contracts have been proposed. At heart, each contract aims to insulate the borrower in the event of area wide house price decline, by suitably adjusting either the outstanding balance or monthly payment of a traditional fixed rate mortgage (FRM). From the bank's perspective, the idea is that if one accounts for foreclosure costs and other negative externalities  associated to underwater default, then, despite the lower payments (compared to the FRM), the contracts are competitive or even preferable.

The goal of our paper is to analyze these proposals, and identify the features which are best for the borrower and lending institution, as well as any unforeseen risks.  We consider the ``adjustable balance mortgage'' of  \cite{ambrose2012adjustable}; the ``continuous workout mortgage'' of \cite{shiller2011continuous}; and the ``shared responsibility mortgage'' of \cite{mian2013state, mian2015house}.\footnote{The shared responsibility mortgage is nearing commercial availability: see \cite{PartnerOwn}.} We choose these contracts because they span a wide range of possible adjustments, such as lowering payments immediately, lowering payments if house prices fall sufficiently far, and including capital gains sharing in high house price states. Following the analysis in \cite{mei2019improving}, we group the above contracts into two broad categories: \emph{adjustable balance mortgages} (ABM) as in \cite{ambrose2012adjustable}; and \emph{adjustable payment rate mortgages} with prepayment penalties (APRM) as in \cite{shiller2011continuous, mian2013state, mian2015house}.

Each contract type modifies the traditional FRM by adjusting payments and balance according to the movements of a (local) house price index\footnote{Local house price indices exist: both the ``S\&P CoreLogic Case-Shiller Home Price Indices'' (\cite{Case_Shiller_Index})  and the Federal Housing Finance Agency House Price Index Reports (\cite{FHFA_HPI}) track national and local house price movements, with Case-Shiller having indices for twenty U.S. metropolitan areas.} $H=\cbra{H_t}$. An index, as opposed to repeated home appraisals, is used for two reasons. First, home appraisals are cumbersome and expensive. Second, and more importantly, using an index removes moral hazard as the borrower should not profit  from intentionally lowering his home value.  With detailed formulas provided in Section \ref{S:mortgages}, we briefly describe the payments at a time $t$ prior to the loan maturity at $T$. For a given mortgage rate $m$, which might vary depending on the contract, set $B^F_t(m)$ as the outstanding balance and $c^F(m)$ the level payment for a traditional FRM with initial loan-to-value of $B_0$ (the purchase price and index value are normalized to $1$ so that $B_0=0.8$ for a $20\%$ down-payment).

The ABM explicitly eliminates the underwater effect by setting the outstanding balance to $\min\bra{B^F_t(m),H_t}$. Therefore, for
a $20\%$ down-payment, house prices would have to fall by $20\%$ before payments are adjusted. The payment rate is derived to be $c^F(m)\times \min\bra{1,H_t/B^F_t(m)}$, so that it also never exceeds the corresponding FRM rate.

Alternatively, the APRM takes as primitive the payment rate $c^F(m)$, and adjusts payments upon any decline in $H$, setting the new rate to  $c^F(m)\times\min\bra{1,H_t}$. The outstanding balance is then derived to be $B^F_t(m)\times \min\bra{1,H_t}$. Additionally, and following the recommendation in \cite{mian2013state}, the APRM has a profit sharing feature, where should the borrower prepay at $t$, he must pay the penalty $\alpha\times\max\bra{H_t-1,0}$, which is $\alpha\times 100\%$ of the capital gains on the house.  In \cite{mian2015house} the authors suggest $\alpha = 0.05$.   The idea is to protect the bank should the borrower refinance into another APRM when house prices are very low (say at $H_l$), because if so-refinanced any future prepayment will incur a large penalty through the capital gains based upon $H_l$.

To identify the contracts' values, we use American options pricing methodology, while also allowing for mortgage turnover\footnote{Turnover refers to non-refinancing mortgage prepayments, such as those due to divorce, job relation, etc..}. More precisely, excluding turnover related prepayments, we assume both locally and globally, the bank takes a worst-case approach.  Locally worst-case means that given a (non-turnover related) termination time (either default or prepayment), the bank assumes it will receive the lower of the two possible payments. Globally worst-case means the bank values the mortgage by considering the worst possible termination time, which is modeled as the optimal stopping time.  This approach avoids explicitly identifying the borrower's rational for default or prepayment, and hence stands in contrast to \cite{mei2019improving}, where the borrower uses utility based-methods for determining her optimal prepayment/default policy in a discrete time setting. However, our continuous time model implicitly assumes a level of financial sophistication for the borrower, but crucially, it allows for  (in the limit of a long contract maturity) \emph{explicit} formulas for the contract values.  This in turn enables us to perform comparative statics analysis and obtain qualitative conclusions. 
 
Applying options pricing theory to value mortgage backed securities is well known: see \cite{kau1992generalized, kau1993option, kau1995overview, stanton1995rational, deng2000mortgage}, as well as \cite{kalotay2004option, chen2009value, de2010optimal,  MR3534469}.  However, it was recognized that borrowers do not always act in a financially optimal manner (see \cite{KEYS2016482}). This led to the popularity of reduced form models for mortgage valuation: see \cite{1981_Dunn_Mcconnell, schwartz1989prepayment,schwartz1992prepayment, kau1995valuation} and their many extensions.  Despite its pitfalls, in order to compare the proposed contracts, we believe the options pricing approach is appropriate.  Simply put, as the contracts' stated objective is to reduce selective default, we must assume the borrower is sophisticated enough to selectively default.

Following the literature (c.f. \cite{ambrose2012adjustable}, \cite{shiller2017continuous}, \cite{mei2019improving}), we assume the discounted house price index $H$ follows a geometric Brownian motion with volatility $\sigma$ and dividend or ``benefit'' rate $\delta$, which measures utility from home ownership. We further assume that the mortgage holder's house price is approximated by the index. While in reality there may be basis risk between the individual house price and the index level, we ignore this risk for reasons of tractability in modeling the true house price process (how often can one observe the ``true'' value? how much does it cost to obtain such a value?).  Additionally, to isolate the relationship between house prices and default, we assume the interest rate is constant. We justify this assumption in Section \ref{S:finite_horizon} by showing in a stochastic interest rate environment that the default boundary is largely insensitive to the interest rate. 
 
In this setting, the contracts' value, as well as prepayment and default option values, are identified with solutions to free boundary problems, which we show in Section \ref{S:perpetual} admit explicit solutions in the case of an infinite mortgage maturity.\footnote{Strategic defaults tend to occur near the beginning of the mortgage term (especially during the financial crisis, \cite{mayer2009rise}), and the typical mortgage contract is 30 years.} Therefore, we can not only can determine the contract and default option values, but also the optimal decision boundaries. We show  the optimal stopping regions for the FRM have two boundaries (one corresponding to prepayment and one to default), while for the ABM there are either one or two thresholds, with each threshold corresponding to prepayment.  Due to the prepayment penalty, for the APRM the stopping region is especially interesting. Indeed, there can be anywhere from one to three stopping boundaries, each attributable to prepayment, depending upon very delicate parameter relationships.

To the best of our knowledge, this is the first paper that makes a comparison of different mortgage contracts in a continuous time model, and  we believe the perpetual, constant interest rate, setting provides an accessible first step in this research direction.  As the following step, one may study the finite maturity contracts (under both constant and stochastic interest rates) from theoretical and numerical perspectives. Especially for the APRM, we expect the optimal stopping regions/boundaries to display an unusual structure, due to the finite time horizon, prepayment penalty $\alpha$, and the kink of the payoff function at $H_0$.

\subsection*{Findings} In the above setting, our main findings are 

\begin{enumerate}[(1)]
\item The APRM capital gain sharing feature is ineffective. Already at low percentages (e.g. $\alpha = 1\%-2\%$), virtually all prepayment when $H>1$ is eliminated, making the contract value insensitive to $\alpha$. As the borrower is essentially locked into her loan when house prices rise, the prepayment penalty will dominate the capital-loss protection, and we envision she will be frustrated with the contract.\footnote{This phenomena occurred in the metropolitan Boston area, where loans offering downside protection were issued to low income buyers, who were unaware of the capital-gain fees upon prepayment (see \cite{globe.20200216}).}  

\item The ABM is competitive with the FRM, while the APRM requires a larger spread. For example, at $35\%$ foreclosure costs (c.f. \cite{hatcher2006foreclosure, campbell2011forced, andersson2014loss}), and a home benefit rate of $7\%$ (c.f. \cite{mei2019improving}) the ABM can offer a lower mortgage rate than the FRM (by about $20$ basis points (bp)) and still yield the same contract value. By contrast, the APRM requires a spread of around $40$ bp (consistent across a range of sharing proportions $\alpha$). 
\item Both the ABM and APRM may endogenously lead to prepayment in low house price states. However, this phenomena does not occur in the presence of either low mortgage rates $m$ or high benefit rates $\delta$  (i.e. if the borrower is  happy living in the house or if barriers to selective default such as rental and/or credit-related costs are high), and hence in practice one does not expect low-state prepayment.
\end{enumerate}



\subsection*{Related Literature}  Despite the disastrous effects of large-scale underwater default, the literature analyzing adjustments to the traditional FRM is scant, especially within the mathematical finance community.  Aside from the papers which introduced the contracts (see \cite{ambrose2012adjustable}, \cite{shiller2011continuous}, \cite{mian2013state} as well the authors' follow up papers), to our knowledge only the recent  \cite{mei2019improving} performed a cross-contract analysis.

There is, however, a distinct strand of literature (see \cite{piskorski2010optimal, piskorski2011stochastic, campbell2011model, eberly2014efficient}  as well as \cite{piskorski2017equilibrium, campbell2018structuring, guren2018mortgage}) which designs optimal mortgage contracts based upon principal-agent and/or equilibrium considerations. Generally, these papers indicate the superiority of an option adjustable rate mortgage (ARM)\footnote{For two exceptions, see \cite{kalotay2015case} which considers an ARM with only negative rate resets, and \cite{piskorski2017equilibrium} which produces a contract similar to the ABM.}, where the borrower is allowed to defer principal payments, leading to negative amortization (and potentially default).  However, option ARMs were actually issued with limited effectiveness (see \cite{gumbinger2010, robertson2013, sialm2016}), because borrowers deferred principal payments, putting themselves at greater risk of default, even if there was no pressing financial need. By contrast, the contracts we consider have automatic payment adjustments, removing the borrower's discretion. Lastly, we highlight \cite{greenwald2018financial}, which argued against adopting the APRM on a large scale (e.g. through FNMA backing) if indexing is done at a national level. However, there is no need to index at the national level, and presently these mortgages are, at most, offered on a very small scale (\cite{PartnerOwn}). As such, we we provide a ``first implementation'' analysis, taking the approach the bank is considering issuing these contracts to small number of sophisticated borrowers, and would like to know how the prepayment/default behavior would change, and how to effectively market the product.

This paper is organized as follows.  Section \ref{S:mortgages} provides continuous time versions of the contracts, while Section \ref{S:general_analysis} formulates the perpetual optimal stopping problems.  Section \ref{S:perpetual} provides theoretical solutions for the value functions and optimal default/prepayment strategies. A numerical analysis is presented in Section \ref{S:numerics}. Section \ref{S:finite_horizon} contains an example in the finite horizon, stochastic interest rate setting, showing how the default boundary is insensitive to the interest rate. We conclude in Section \ref{S:conclusion}. Proofs are in Appendix \ref{AS:S4_proofs}.

\section{\bf The Mortgages}\label{S:mortgages}

We first introduce the mortgages. Each involves a loan of $B_0$ at time $0$ with maturity $T$. We normalize the purchase price so that $B_0$ is the initial loan to value (LTV). We do not assume $B_0=1$ as typical initial LTVs are $0.8$ (for a $20\%$ down-payment) or $0.9$ (for a $10\%$ down-payment). Additionally, there is a house price index process $H = \cbra{H_t}_{t\geq 0}$ with $H_0 = 1$ also normalized.

\subsection{Fixed rate mortgage} The baseline contract is a continuous time, fully amortized, level payment FRM with mortgage rate $m^F$.  The  outstanding balance $B^F$ is given by
\begin{equation}\label{E:frm_balance}
B^F_t = \frac{B_0\left(1-e^{-m^F(T-t)}\right)}{1-e^{-m^FT}}
\end{equation}
for $t\leq T$ and the continuously paid coupon rate $c^F$ is
\begin{equation}\label{E:frm_coupon}
c^F = \frac{m^FB_0}{1-e^{-m^FT}} = \frac{m^F B^F_t}{1-e^{-m^F(T-t)}}.
\end{equation}
 If the homeowner decides to prepay at $t\leq T$, the bank receives
\begin{equation*}
\B^F_t \dfn  B^{F}_t.
\end{equation*}
However, the FRM  allows for the mortgage to be underwater, i.e., $H_t<B^F_t$, and hence creates incentive for strategic default.

\subsection{Adjustable Balance}  The ABM adjusts the outstanding balance in the event of sufficiently large house price decline, and then the payment rate is updated based upon the adjusted balance. The ABM was proposed in \cite{ambrose2012adjustable}, and here we present a continuous time version.  It starts with a mortgage rate $m^A$, and computes a nominal remaining balance $\wh {B}^{A}$ and payment rate $\wh {c}^A$ using \eqref{E:frm_balance} and \eqref{E:frm_coupon}, respectively, but with $m^A$ replacing $m^F$ therein. Then, the actual remaining balance $B^A$ is set to
\begin{equation}\label{E:abm_balance}
B^{A}_t \dfn \min\bra{\wh{B}^A_t, H_t}
\end{equation}
for $t\leq T$. To compute the actual payment rate $c^{A}$, assume that at time $t$ the homeowner has borrowed $B^A_t$ in a fixed rate, level payment, loan with maturity $T-t$ and contract rate $m^A$.  From \eqref{E:frm_coupon} and \eqref{E:abm_balance} we deduce
\begin{equation*}
c^{A}_t = \frac{m^AB^{A}_t}{1-e^{-m^A(T-t)}} = \wh {c}^A\times \min\bra{1, \frac{ H_t}{\wh{B}^A_t}}
\end{equation*}
for $t\leq T$.  For $20\%$ down payment, house prices would have to drop by at least $20\%$ before the ABM adjusts payments.  By design, the ABM is never underwater, as $B^A_t \leq H_t$ and upon prepayment at $t\leq T$ the bank receives
\begin{equation*}
\B^A_t \dfn B^A_t = \min\bra{\wh{B}^A_t,  H_t}.
\end{equation*}

\subsection{Adjustable Payment Rate} The APRM adjusts the payment rate based upon the house price index, and then the balance is computed from the adjusted payment rate. The two canonical examples are the continuous workout mortgage (CWM) as proposed in \cite{shiller2011continuous} (and subsequently analyzed in \cite{shiller2013mitigating, shiller2017continuous}) and the shared responsibility mortgage (SRM) from \cite{mian2013state} (as well as \cite{mian2015house, EquitableGrowth}). Like the ABM, the APRM starts with a mortgage rate $m^P$ and computes the nominal balance $\wh{B}^P$ and payment rate $\wh{c}^P$ according to \eqref{E:frm_balance} and \eqref{E:frm_coupon} (with $m^P$), respectively.  However, unlike the ABM, the APRM adjusts the payment rate upon \emph{any} decline in the house price index, setting the true payment rate to
\begin{equation*}
c^{P}_t \dfn \wh{c}^P \times \min\bra{1,H_t}
\end{equation*}
for $t\leq T$.  Then the remaining balance is set to
\begin{equation*}
B^{P}_t \dfn c^{P}_t \times \frac{1-e^{-m^P(T-t)}}{m^P} = \wh{B}^P_t \times \min\bra{1,H_t}
\end{equation*}
 for $t\leq T$. Additionally, as proposed in \cite{mian2013state}, in order to compensate the bank for the reduced payment rate, the APRM has a high-state prepayment penalty, which requires the borrower to share the portion $\alpha (H_t-1)^+$ (recall $H_0= 1$)  of the capital gain with the lender, where $0\leq \alpha < 1$. Therefore, upon prepayment at  $t<T$ the bank receives
\begin{equation*}
\B^P_t \dfn B^{P}_t  + \alpha\left(H_t - 1\right)^+ = \wh{B}^P_t\min\bra{1,H_t} + \alpha(H_t - 1)^+.
\end{equation*}
In \cite{mian2015house}, $\alpha=5\%$ is recommended, but we will leave $\alpha$ general in order to study the sensitivity of the contract value to $\alpha$.  However, to avoid mathematical technicalities which arise for large $\alpha$, we assume throughout that
\begin{ass}\label{A:aprm_alpha} $\alpha < B_0$.
\end{ass}
From a financial perspective this is not even an assumption. Indeed, typical initial loan to value ratios are $B_0 = .8$ or $B_0 = .9$, where as the suggested value of $\alpha$ in \cite{mian2015house} was $\alpha = .05$, and borrowers would never accept a contract which taxed away $80-90\%$ of the profits on the house. 


\section{\bf Perpetual Contracts: Optimal Stopping Problems and Option Values}\label{S:general_analysis}

For finite horizons, using American options methodology to value the contracts yields  free boundary problems for which closed from solutions are not available. Hence, to qualitatively compare the contracts,  we assume an infinite horizon.  Each of the balances, payment rates, and prepayment amounts are easily derived from their finite maturity analogs by taking $T=\infty$.

\begin{equation}\label{E:perpetual}
\begin{small}
\begin{tabular}{c||c|c}
Contract & Payment Rate at $t$ & Prepayment Amount at $t$ \\
\hline
FRM &  $m^FB_0$ & $B_0$\\
\hline
ABM & $m^A\min\bra{B_0, H_t}$ & $\min\bra{B_0,H_t}$\\
\hline
APRM & $m^P B_0\min\bra{1,H_t}$ & $B_0\min\bra{1,H_t} + \alpha(H_t-1)^+$
\end{tabular}
\end{small}
\end{equation}
\medskip

As the typical maturity is $30$ years,  the perpetual assumption is not strong, given that selective default decisions typically occur near the beginning of the term  \cite{mayer2009rise, chan2016determinants}. Also as our focus is on default incentives, we believe it is reasonable to assume a constant interest rate. This assumption is justified in Section \ref{S:finite_horizon} where we show in a stochastic rate environment that the default boundary is rather insensitive to the interest rate level.

We now formulate the infinite horizon optimal stopping problems associated to the contract values as well as the respective default option values. 

\subsection{Model and Assumptions} There is a filtered probability space $\basisq$, where $\qprob$ is the risk-neutral pricing measure. The risk-free rate is $r>0$.  The discounted house price index, in accordance with the literature (\cite{ambrose2012adjustable}, \cite{shiller2017continuous}, \cite{mei2019improving}) follows a geometric Brownian motion (GBM) with constant dividend rate $\delta>0$, which captures the benefit rate (either monetarily or through utility of home-ownership) of residing in the home
\begin{equation}\label{E:H_dynamics}
\frac{dH_t}{H_t} = (r -\delta)dt + \sigma dW_t,\qquad H_0 = 1.
\end{equation}
Here, $W$ is a $\qprob$-Brownian motion and $ \sigma >0$ is the constant volatility parameter.  Next, we make precise our behavioral/structural assumptions, as subtleties arise in the use of risk-neutral pricing for mortgage backed securities which, though common in the literature (c.f. \cite{kau1995valuation, kalotay2004option, chen2009value, de2010optimal, MR3534469} amongst many others), should be made explicit.  Of particular note, in contrast to \cite{campbell2003household, campbell2015model, mei2019improving} we take the bank's perspective on issuing the contract, and do not directly model the homeowner's behavior, except as we will see, through the assumptions implicit in the bank's worst case approach.

The bank recognizes that borrowers may prepay or default at both strategic and non-strategic times, with non-strategic times corresponding to turnover (i.e prepayment/default due to income loss, job relocation, death, divorce, etc.). However, at both non-strategic and strategic times, the bank assumes the borrower will do what is  worst for the bank.  More precisely, we set $\tau_{to}$ as the turnover time, and assume $\tau_{to}$  has constant $\filt^W$  intensity  $\lambda$.\footnote{Technically, $\tau_{t0}$ is the first jump time of an independent Poisson process $N_{to}$ with rate $\lambda$, also defined on $\basisq$. Then, the filtration $\filt$ is generated by $W$ and the process $t\to 1_{\tau_{to}\geq t}$.} Next, we set $\tau$ as a strategic default time  (i.e. $\tau$  is a $\filt^W$ stopping time). At both $\tau_{to}$ and $\tau$ the bank assumes the borrower will do what is worst for the bank. If the borrower prepays the bank receives the remaining balance (for whichever contract is being used).  If the borrower defaults, the bank receives the house value.\footnote{In Section \ref{SS:foreclosure} we account for foreclosure costs, which significantly reduce the amount the bank receives upon default. Indeed, foreclosure can take up to $3$ years (\cite{NY_Foreclosure}), with total costs (due to maintenance, marketing and discounted ``fire sale'' pricing) approaching $35-40\%$ of the home value: see \cite{ambrose1996cost, hatcher2006foreclosure, campbell2011forced, andersson2014loss}.}  Therefore, taking a worst-case perspective, the bank assumes it will receive $\min\bra{H,\B}$ with prepayment when $\B< H$ and default when $H\leq \B$.


Continuing, we assume the bank has access to a liquid market which trades in $H$ and a money market account with rate $r$, and uses $\qprob$ as a risk-neutral pricing measure. We stress that we are not assuming a liquid market for the borrower's home price. Rather, the liquid market is for the home price index, and our assumption is that the borrower's home price is well-approximated by the index level. Furthermore, the bank assumes the borrower will do what is strategically worst, and values the mortgage by minimizing the expected discounted payoff over all termination (stopping) times.

Importantly, we do not examine the borrower's rationale for prepaying (i.e. refinancing versus selling) or defaulting. However, in the absence of frictions (e.g. foreclosure costs, refinancing costs, moving costs), there is a direct connection between the bank applying a worst case analysis, and assuming the borrower is a financial optimizer. By contrast, when incorporating frictions such as foreclosure costs for the bank (see Section \ref{SS:foreclosure}), this connection is not as strong, as these costs do not factor into a homeowner's decision to default.

\subsection{Mortgage, Default Option Values}
Having stated the model, we next define the contract/option values as the value functions of corresponding optimal stopping problems. 
For $i\in\cbra{F,A,P}$, recall that $c^{i}$ is the cash flow rate, and $\B^i$ the prepayment amount (remaining balance plus possible penalty) should prepayment occur. Let us consider a time  $t\geq 0$ and assume that neither prepayment nor default has occurred by $t$. Consistent with the previous section's discussion, the bank assigns the contract a value of
\begin{equation*}
    \begin{split}
        V^i_t \dfn &\underset{\tau\ge t}{\operatorname{inf}}\ \condexpv{}{}{1_{\tau_{to} > t} \left(\int_t^{\tau\wedge\tau_{to}} e^{-r(u-t)}c^i_u du + e^{-r(\tau\wedge\tau_{to} - t)}\min\bra{H_{\tau\wedge\tau_{to}},\B^i_{\tau\wedge\tau_{to}}}\right)}{\F_t}\\
        =&1_{\tau_{to} > t}\underset{\tau\ge t}{\operatorname{inf}}\ \condexpv{}{}{\int_t^{\tau} e^{-(r+\lambda)(u-t)}\left(\lambda\min\bra{H_u,\B^i_u} + c^i_u\right)du + e^{-(r+\lambda)(\tau - t)}\min\bra{H_\tau,\B^i_\tau}}{\F^W_t}
    \end{split}
\end{equation*}
where we used the independence of $W$ and $\tau_{to}$, and that the latter is exponentially distributed. Above, the infimum is taken with respect to $\filt^W$ stopping times exceeding $t$.  Therefore, on the set $\cbra{\tau_{t0} > t, H_t = h}$ and using the Markov property for $H$, the bank assigns the contract the value $V^i_t = V^I(t,h)$ where
\begin{equation}\label{E:fin_horizon_vf}
\begin{split}
V^{i}(t,h) &\dfn \underset{\tau\ge t}{\operatorname{inf}}\ \espalt{}{t,h}{\int_t^\tau e^{-(r+\lambda)(u-t)}\left(\lambda\min\bra{H_u,\B^i_u} + c^i_u\right) du + e^{-(r+\lambda)(\tau - t)}\min\bra{H_\tau,\B^i_\tau}}\\
&= V^{i}(0,h).
\end{split}
\end{equation}
for $h>0$ and where we have written $\espalt{}{t,h}{\cdot}$ for $\condexpv{}{}{\cdot}{H_t=h}$. The last equality (i.e., time-independence of $V^i$) follows from \eqref{E:perpetual} and the time-homogeneity of $H$. Recall that we assumed $H_0=1$ when defining the contracts, hence when $t=0$ we are interested in the value of $V^i(0,h)$ at $h=1$ only. But for $t>0$, the house price can take any value $h>0$ and, in what follows, we will analyze the behavior of the mapping $h\mapsto V^i (t,h)$.  To understand the above formula, in particular the role played by turnover, note that for a given termination time $\tau$, the expectation is the arbitrage-free price for the cash flow $c^i + \lambda \min\bra{H,\B^i}$ until $\tau$, followed by a lump-sum payment of $\min\bra{H_\tau,\B^i_\tau}$ at $\tau$, provided we discount at the higher rate $r+\lambda$. Then, the mortgage value is found by applying the worst-case analysis over all such stopping times.

We next turn to the default option value, which  is the cost incurred by the bank due to the fact that the borrower can both default and prepay, rather than just prepay. While accounting for both prepayment and default yields the payment $\min\bra{H,\B^i}$,  excluding default gives the larger payment $\B^i$. Additionally, there is no reason to think the strategic worst-case termination time accounting for both prepayment and default is the same worst case termination time accounting only for prepayment. Accordingly, following similar computations which lead to \eqref{E:fin_horizon_vf} on the set $\cbra{\tau_{to}>t, H_t = h}$ for $i\in\cbra{F,A,P}$, we define the value of mortgage excluding defaults as
\begin{equation}\label{E:mort_val_primal_no_default} 
\begin{split}
V^{NoDef,i}(t,h) &\dfn \underset{\tau\ge t}{\operatorname{inf}}\espalt{}{t,h}{\int_t^\tau e^{-(r+\lambda)(u-t)}\left(\lambda \B^i_u + c^i_u\right) du + e^{-(r+\lambda)(\tau - t)}\B^i_\tau}.
\end{split}
\end{equation}
The default option cost to the bank is $D^{i} \dfn  V^{NoDef,i} - V^{i}$.


\section{\bf Perpetual Contracts: Solution}\label{S:perpetual}

We now provide explicit solutions to the stopping problems given in the previous section. In addition to an infinite horizon and constant interest rate, we also assume $ m^{i}>r$, $i\in\cbra{F,A,P}$ because negative mortgage spreads are unrealistic, and we recall from Assumption \ref{A:aprm_alpha} that $\alpha < B_0$.

\subsection{Free Boundary Problems}\label{SS:stopping}

It is standard procedure in the American option pricing literature to reduce perpetual optimal stopping problems to free boundary ODE systems. Furthermore, as we are in a time-homogeneous Markovian setting, we can focus on $t=0$, where as shown above, each of the perpetual contract values  takes the form 
\begin{equation}\label{E:abstract_v}
     V(h) = \underset{\tau\geq 0}{\operatorname{inf}}\espalt{}{h}{\int_0^\tau e^{-(r+\lambda)u}c(H_u) du + e^{-(r+\lambda)\tau}f(H_\tau)}
\end{equation}
for certain functions $c(h),f(h)$ of the house price. Indeed, one has
\begin{equation}\label{E:perpetual_flow_payoff}
\begin{small}
\begin{tabular}{c||c|c}
Contract & $c(h)$ & $f(h)$ \\
\hline
FRM &  $m^F B_0 +\lambda\min\bra{B_0,h} $ & $\min\bra{B_0,h}$\\
\hline
ABM & $(m^A+\lambda)\min\bra{B_0,h}$ & $\min\bra{B_0,h}$\\
\hline
APRM & $(m^P+\lambda)B_0\min\bra{1,h} + \alpha\lambda (h-1)^+$ & $B_0\min\bra{1,h} + \alpha(h-1)^+$
\end{tabular}
\end{small}
\end{equation}
\medskip
\begin{rem}\label{R:ABM_APRM_no_default_value}
In \eqref{E:perpetual_flow_payoff}, we used that for all times $u\geq 0$
\begin{equation*}
    \begin{split}
        \textrm{(ABM)}\quad & \min\bra{\B^A_u,H_u} = \min\bra{\min\bra{B_0,H_u},H_u} = \min\bra{B_0,H_u} = \B^A_u;\\
        \textrm{(APRM)}\quad & \min\bra{\B^P_u,H_u} = \min\bra{B_0\min\bra{1,H_u} + \alpha(H_u-1)^+,H_u}\\
        & = B_0\min\bra{1,H_u} + \alpha(H_u-1)^+ = \B^P_u.
    \end{split}
\end{equation*}
The first equality is clear, while the second follows from straightforward calculations using $B_0 \leq 1$.   Therefore, as the remaining balance for both the ABM and APRM is always dominated by the house price, we see that strategic ruthless  default is by definition eliminated for these two contracts.  As such, the default option is $0$ for these two contracts.
\end{rem}
The free boundary ODE associated to the optimal timing problem in \eqref{E:abstract_v} is 
\begin{equation}\label{E:fbp}
\min\bra{\Lcal_{H} V - (r+\lambda) V + c, f-V}(h) = 0;\qquad  h>0
\end{equation}
where 
\begin{equation}\label{E:H_op}
    \Lcal_{H}  =  (r-\delta)h \frac{\partial }{\partial h}+(1/2)\sigma^2 h^2 \frac{\partial^2 }{\partial h^2}
\end{equation}
is the second order operator associated to $H$. The exercise region is $E\dfn \cbra{V=f}$, while the continuation region is $C\dfn \cbra{V<f}$. These regions must be determined, along with the solution $V$ to the ODE.   

As usual, continuous and smooth pasting conditions at optimal stopping boundaries are imposed to obtain $C^1$ solutions amenable to \ito's formula and hence verification (c.f. \cite{MR2167640} for extension of \ito's formula to $C^1$ functions).  We connect solutions to \eqref{E:abstract_v} and \eqref{E:fbp}  using the following verification result, proved in Appendix \ref{AS:S4_proofs}.

\begin{prop}\label{P:verification}
Let $V:(0,\infty)\to\reals$ be $C^1$, with only a finite number of points where $V$ is not $C^2$,\footnote{Technically, for some $-\infty = a_0 < a_1 < a_2 < ... < a_N < a_{N+1}=\infty$, $V$ is $C^2$ on each $(a_n,a_{n+1})$, $n=0,\dots,N$.} and such that for all $h>0$
\begin{equation}\label{E:dotV_integ}
    \espalt{}{h}{\int_0^\infty e^{-2(r+\lambda)u}H_u^2 V_h(H_u)^2 du} < \infty;\qquad \lim_{t\to\infty} e^{-(r+\lambda)t}V(H_t) \textrm{ exists },\ \prob^h \textrm{ almost surely}.
\end{equation}
Let $c,f:(0,\infty)\to\reals$ be continuous non-negative functions, and assume $(0,\infty) = C \cup E$ where $C$ is a finite union of open intervals, and $f\in C^2(U)$ for some open $U \supset E$. Assume further that $\lim_{t\to\infty} e^{-(r+\lambda)t} f(H_t)$ exists $\prob^h$-a.s. for all $h>0$. Then, if $V$ satisfies \eqref{E:fbp} with $\Lcal_H V-(r+\lambda)V + c = 0$ in $C$ and $V=f$ in $E$, $V$ solves the optimal stopping problem in \eqref{E:abstract_v} and the optimal stopping time is $\tau^* = \inf\cbra{t\ge 0\such H_t \in E}$.
\end{prop}

The key observation is that the homogeneous ODE $\Lcal_H V - (r+\lambda)V  =  0$ has general solution
\begin{equation*}
V(h)= C_1h^{p_1} + C_2h^{-p_2}
\end{equation*}
where  $C_1$ and $C_2$ are free constants, and
\begin{equation}\label{E:p1_p2}
\begin{split}
p_1 &= -\frac{r-\delta-\sigma^2/2}{\sigma^2} + \frac{1}{\sigma^2}\sqrt{\left(r-\delta-\sigma^2/2\right)^2 + 2(r+\lambda)\sigma^2} > 1;\\
p_2 &= \frac{r-\delta-\sigma^2/2}{\sigma^2} + \frac{1}{\sigma^2}\sqrt{\left(r-\delta-\sigma^2/2\right)^2 + 2(r+\lambda)\sigma^2} > 0.\\
\end{split}
\end{equation}
We use this fact to obtain explicit solutions to \eqref{E:fbp} for the FRM, ABM, APRM.  Before stating our results, we present a very useful identity used repeatedly throughout
\begin{equation}\label{E:p1_p2_ident}
\frac{1+p_2}{p_2}\times\frac{p_1-1}{p_1} = \frac{\delta+\lambda}{r+\lambda}.
\end{equation}

\subsection{FRM} Let us first consider the FRM where from \eqref{E:perpetual_flow_payoff} we see the functions $c$ and $f$ clearly satisfy the hypotheses of Proposition \ref{P:verification} provided $B_0$ is not a boundary point between the continuation and exercise regions (which we will show).  In the continuation region, the general solution to  $\Lcal_H V^F(h) - (r+\lambda)V^F(h) + c(h)  = 0$ is
\begin{equation*}
    \begin{split}
    V^F(h) &= \begin{cases} C_1 h^{p_1} + C_2 h^{-p_2} + \frac{m^F}{r+\lambda}B_0 + \frac{\lambda}{\lambda+\delta}h, & h < B_0\\  \wt{C}_1 h^{p_1} + \wt{C}_2 h^{-p_2}  + \frac{m^F + \lambda}{r+\lambda} B_0, & h > B_0\end{cases}
    \end{split}
\end{equation*}
for to-be-determined constants $C_1,C_2,\wt{C}_1,\wt{C}_2$. Here, we will identify two boundaries $h_1 < B_0 < h_2$ such that default occurs optimally for $h\leq h_1$, prepayment occurs optimally for $h\geq h_2$, and continuation occurs within. The  solution is obtained by finding  $(C_1,C_2,\wt{C}_1,\wt{C}_2,h_1,h_2)$ such that $V^F$ satisfies the continuous and smooth pasting conditions at $h_1,B_0,h_2$, and $C_1,C_2,\wt{C}_1,\wt{C}_2<0$ so that $V^F$ is concave with $V^F(h)\leq \min\bra{h,B_0}$. The following proposition, the proof of which is given in Appendix \ref{AS:S4_proofs}, summarizes the solution.

\begin{prop}\label{P:FRM_P}  For $V^{F}$ has the following action regions and  value function 
\begin{equation*}
\begin{tabular}{c||c|c|c|c}
    $h$ & $\le h_1$ & $\in (h_1, B_0)$ & $\in (B_0, h_2)$ & $\geq h_2$ \\
    \hline
    Action & Default & Continue & Continue & Prepay \\
    \hline
    $V^F(h)$ & $h$ & $C_1 h^{p_1} + C_2 h^{-p_2} + \frac{m^F}{r+\lambda}B_0 + \frac{\lambda}{\lambda+\delta}h$ & $\wt{C}_1 h^{p_1} + \wt{C}_2 h^{-p_2}  + \frac{m^F + \lambda}{r+\lambda} B_0$ &$B_0$\\
\end{tabular}
\end{equation*}
where  $0 < h_1 < B_0 < h_2$ and $C_1,C_2,\wt{C}_1,\wt{C}_2 < 0$ (see Figure \ref{F:value_frm} for illustrations).
\end{prop}

\begin{figure}[t]
\includegraphics[height=5cm,width=7cm]{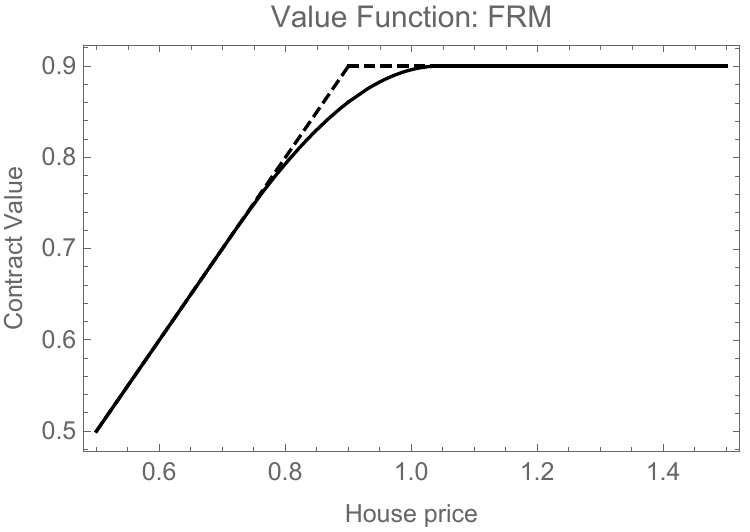}\\
\caption{ The value function $V^F(h)$ of FRM (solid line) versus the payoff function $f(h)=\min(B_0,h)$ (dashed line). The optimal boundaries are $h_1=0.72$ and $h_2=1.04$. The parameters are given in Table 1 with $m^F=m_0$ and $\delta=0.03$. }
\label{F:value_frm}

\end{figure}

\begin{rem} As with all our results in this section, the constants $C_1,C_2,\wt{C}_1,\wt{C}_2$ and boundaries $h_1,h_2$ admit explicit or semi-explicit solutions  (e.g. as the roots to algebraic equations). Also, note that if $\lambda=0$, then $C_1=\wt{C}_1$, $C_2=\wt{C}_2$, and $V^F(h)=C_1 h^{p_1} + C_2 h^{-p_2} + \frac{m^F}{r}B_0$ for $h\in(h_1,h_2)$.
\end{rem}

To compute the default option value, we must identify $V^{NoDef,F}$  from \eqref{E:mort_val_primal_no_default}. As this contract is artificial,  used only to isolate the value of default, we will not use the term ``prepayment''.  Rather we will use ``stop''. Note that excluding default, the respective functions $c,f$ in \eqref{E:fbp} are $c(h) = (m^F+\lambda) B_0$, $f(h) = B_0$.  These satisfy the assumptions of Proposition \ref{P:verification} because $e^{-(r+\lambda)t}H_t = e^{-(\delta+\lambda+\sigma^2/2)t + \sigma W_t} \rightarrow 0$ as $t\to\infty$.   As for the value function, we have the following result.

\begin{prop}\label{P:FRM_P_option}
For $V^{NoDef,F}$ immediate stopping is optimal and $V^{NoDef,F}(h) = B_0$ for all $h>0$. 
\end{prop}


\subsection{ABM}  Next we consider the ABM where $c,f$ are given in \eqref{E:perpetual_flow_payoff}. Here, in the continuation region, $\Lcal_H V^A(h) - (r+\lambda)V^A(h) + c(h)  = 0$ has solution
\begin{equation*}
V^A(h)=  \begin{cases} C_1h^{p_1} + C_2h^{p_2} + \frac{m^A+\lambda}{\delta+\lambda} h, & h<B_0\\ \widetilde{C}_1h^{p_1} + \widetilde{C}_2 h^{p_2} + \frac{m^A+\lambda}{r+\lambda}B_0, & h> B_0\end{cases}
\end{equation*}
where $C_1,C_2,\widetilde{C}_1,\widetilde{C}_2$ are free constants to be determined  together with the optimal prepayment threshold/s using boundary conditions. We recall that default is explicitly ruled out for the ABM, while in the prepayment region $V^A(h) = \min\bra{B_0, h}$. The next proposition characterizes the value function, showing the (surprising) existence of a prepayment region in low housing states, at least when the utility from occupying the house is sufficiently low. Figure \ref{F:value_abm}  illustrates the result.

\begin{prop}\label{P:ABM_P} The value function $V^{A}$ is increasing, $C^1$ and concave. 

(i) When $m^A\leq \delta$, $V^A$ has action regions
\begin{equation*}
    \begin{tabular}{c||c|c|c}
    $h$ & $\le B_0$ & $\in (B_0, h_2)$ & $\ge h_2$ \\
    \hline
    Action & Continue & Continue & Prepay \\
    \hline
    $V^A(h)$ & $C_1h^{p_1}  + \frac{m^A+\lambda }{\delta+\lambda} h$ & $\widetilde{C}_1h^{p_1} + \widetilde{C}_2h^{p_2} + \frac{m^A+\lambda}{r+\lambda}B_0$ & $B_0$\\
\end{tabular}
\end{equation*}
where $h_2$ is the optimal prepayment boundary. 

(ii) When $m^A>\delta$, $V^A$ has action regions
\begin{equation*}
\begin{tabular}{c||c|c|c|c}
    $h$ & $\le h_1$ & $\in (h_1,B_0]$ & $\in [B_0,h_2)$ & $\ge h_2$ \\
    \hline
    Action & Prepay & Continue & Continue & Prepay \\
    \hline
    $V^A(h)$ & $h$ & $C_1 h^{p_1} + C_2 h^{p_2} + \frac{m^A+\lambda}{\delta+\lambda} h$ & $\widetilde{C}_1h^{p_1} + \widetilde{C}_2h^{p_2} + \frac{m^A+\lambda}{r+\lambda}B_0$ & $B_0$\\
\end{tabular}
\end{equation*}
where $h_1$ and $h_2$ are the optimal prepayment thresholds.\footnote{$C_1,\wt{C}_1,\wt{C}_2,h_2$ for $m^A \leq \delta$ need not coincide with $C_1,\wt{C}_1,\wt{C}_2,h_2$ for $m^A >\delta$.}
\end{prop}

\begin{figure}[t]
\includegraphics[height=5cm,width=7cm]{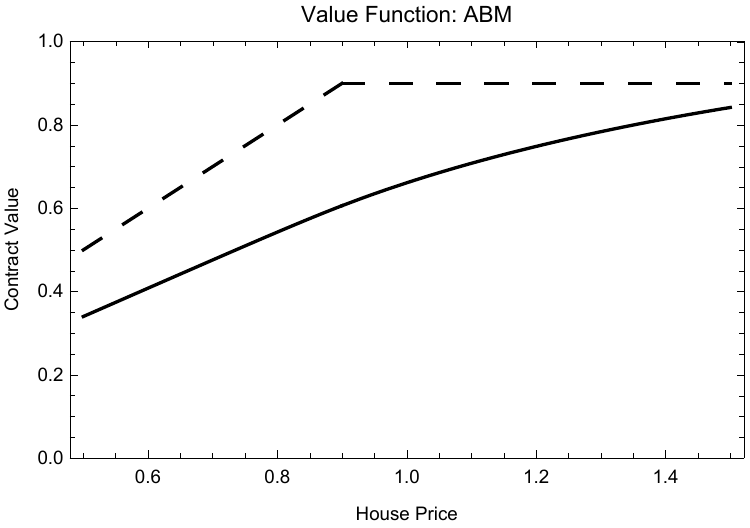}\quad
\includegraphics[height=5cm,width=7cm]{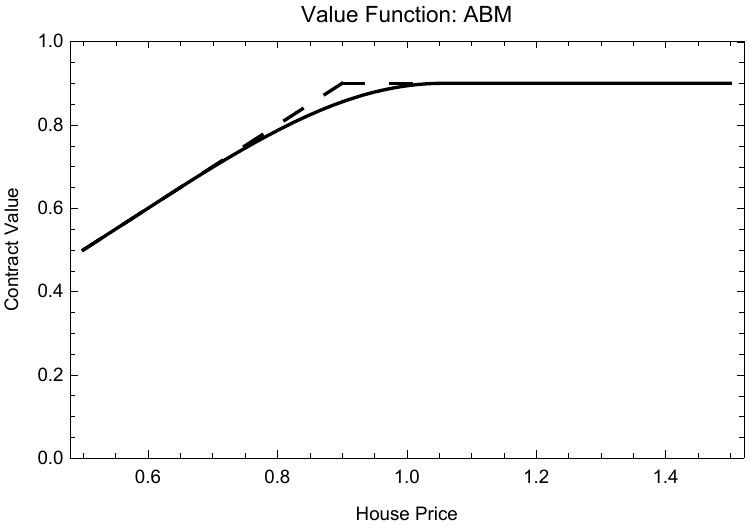} \\

\caption{ The ABM value function $V^A(h)$ (solid) versus the payoff function $\min(B_0,h)$ (dashed). Left panel: $m^A < \delta = 7\%$, with optimal boundary $h_2=1.88$. Right panel: $m^A > \delta = 3\%$, with optimal boundaries $h_1=0.64$ and $h_2=1.06$. All other parameters are as in Table \ref{tab:parameters} with $m^A = m_0$ therein.}
\label{F:value_abm}
\end{figure}

Let us give some intuition for why there is a ``lower'' prepayment region when $m^A>\delta$. When $h<B_0$, if the borrower prepays and sells the house he receives $h-\min\bra{B_0,h} = 0$.  Conversely, by continuing, on the net, he instantaneously pays  $h(m^A-\delta) dt$ where we take into account the utility flow $\delta h\,dt$. Thus, he has an incentive to prepay.  Of course, by prepaying the borrower is giving up the opportunity to prepay in the future,  but when the current home price $h$ is very low, the future prepayment is of lesser value. This is why prepayment occurs only when house price falls below some critical threshold $h_1<B_0$. To obtain the six unknowns $(h_1,h_2,C_1,C_2,\tilde{C_1},\tilde{C_2})$, in this case we impose both continuous and smooth pasting conditions at $h_1$, $B_0$ and $h_2$, i.e., providing six equations.

When $m^A\leq \delta$, the instantaneous net payment flow $ h(m^A-\delta) dt$ is non-positive so there is no prepayment region below $B_0$. More formally, prepaying yields $\min\bra{h,B_0} = h$ which is sub-optimal, as continuing forever yields the lower value
\begin{equation*}
\begin{split}
    &\espalt{h}{}{\int_0^\infty e^{-(r+\lambda)u} (m^A+\lambda)\times \min\bra{B_0, H_u} du}\\
    &\qquad < (m^A+\lambda) \int_0^\infty e^{-(r+\lambda)u}\espalt{h}{}{H_u}du = \frac{m^A+\lambda}{\delta+\lambda} h \leq h.
\end{split}
\end{equation*}
To determine uniquely the four unknowns $(h_2,C_1,\tilde{C_1},\tilde{C_2})$, we impose both continuous and smooth pasting conditions at $B_0$ and $h_2$.

At first glance, the ABM low prepayment region and FRM default region appear similar. However, there is an important difference.  For the FRM, the borrower is defaulting, which induces significant foreclosure costs to the bank. For the ABM the borrower is not defaulting, rather she is refinancing, or selling the home.  Her desire to prepay is based primarily on cash flow considerations.

That the low prepayment region disappears when $m^A\leq \delta$ provides a key insight into the value of the ABM (and, as we will see, the APRM as well). Having removed the default incentive, the homeowner will remain in the mortgage provided his utility is high enough in comparison to the interest he pays. Especially when this utility is high (e.g he likes the neighborhood or house; rents are expensive; he does not want the negative credit associated to default) the borrower will not prepay at low values, and the bank will not receive the house value in the depressed state. 

Lastly, from Remark \ref{R:ABM_APRM_no_default_value} we see that the ABM eliminates the default option value and hence $V^{NoDef,A} \equiv 0$. 



\subsection{APRM} We lastly consider the APRM where $c,f$ are given in \eqref{E:perpetual_flow_payoff}. In the continuation region, a particular solution to $\Lcal_H V^P(h)-(r+\lambda) V^P(h) + c(h) = 0$ is
\begin{equation*}
V^{P}_{par}(h)=\begin{cases} \frac{(m^P+\lambda)B_0}{\delta+\lambda}h, & h < 1\\
    \frac{\alpha\lambda}{\delta +\lambda}h + \frac{(m^P+\lambda)B_0-\alpha\lambda}{r+\lambda}, & h > 1\end{cases}
\end{equation*}
and the general solution is
\begin{equation*}
V^P(h)= \begin{cases} C_1 h^{p_1} + C_2 h^{-p_2} + V^{P}_{par}(h), & h < 1\\
    \wt{C}_1h^{p_1} + \wt{C}_2 h^{-p_2} + V^{P}_{par}(h), & h > 1\end{cases}.
\end{equation*}
Before formally presenting the results, we would like to explain what is happening, as the sharing proportion $\alpha$ complicates matters.

First, as with the ABM, if $m^P\leq \delta$ there is no low-state prepayment region and the reasoning for this is the same as for the ABM. Next, there is a threshold $\alpha^*$ such that a high state prepayment region (i.e. contained in $(1,\infty)$) emerges only if $\alpha < \alpha^*$. Intuitively this is clear, as for $\alpha$ high enough the penalty renders any prepayment benefit moot. Quantitatively, this is justified in the following remark, which also motivates our presentation of results.

\begin{rem}\label{R:alpha_star}

Consider when $h>1$. To simplify the presentation, we will assume no turnover ($\lambda = 0$), which in view of Remark \ref{R:turnover_APRM} and Lemma \ref{L:APRM_turnover_adjustment} below, entails essentially no loss of generality. From \eqref{E:abstract_v} and \eqref{E:perpetual_flow_payoff} we see that continuing forever ($\tau\equiv \infty$) yields a value dominated by $m^P B_0/r$, while prepaying immediately ($\tau \equiv 0$) gives $B_0 + \alpha(h-1)$.  As such, for any $\alpha > 0$, if  $h$ is large enough then prepayment is never optimal, and this leads one to ask if there is an $\alpha^*$ such that $\alpha \geq \alpha^*$ implies it is \emph{never} optimal to prepay when $h>1$.

Let us first assume $\alpha$ yields a high state prepayment region. Thus, there is an $h^* > 1$ such that for $h> h^*$ continuing is optimal (giving $V^P(h) = \wt{C}_2 h^{-p_2} + m^P B_0 /r$), while for $h$ immediately below $h^*$ prepayment is optimal (giving $V^P(h) = B_0 + \alpha(h-1)$). Value matching and smooth pasating at $h^*$ yield
\begin{equation*}
h^* = \frac{p_2}{1+p_2}\left(\frac{B_0}{\alpha}\left(\frac{m^P}{r}-1\right) + 1\right).
\end{equation*}
But, $h^*>1$ implies  $\alpha < p_2B_0(m^P/r-1)$ so for $\alpha \geq p_2B_0(m^P/r-1)$ there is no high state prepayment region. Next, recall from Assumption \ref{A:aprm_alpha} that 
\begin{equation}\label{E:mstar_first_use}
    \alpha < B_0 = p_2B_0\left(\frac{m^P}{r}-1\right) - \frac{p_2 B_0}{r}\left(m^P - m^{*}\right);\qquad m^{*} \dfn \frac{(1+p_2)r}{p_2} = \frac{p_1\delta}{p_1-1},
\end{equation}
where the last equality uses \eqref{E:p1_p2_ident}. Therefore, if $m^P\geq m^*$ then $h^* > 1$ and we may not (yet)  rule out the existence of a high state prepayment region.   However, for $m^P < m^*$, if 
\begin{equation*}
   \alpha \in \left(p_2B_0\left(\frac{m^P}{r}-1\right), B_0\right)
\end{equation*}
then there is no high state prepayment region. For  $0 < \alpha < B_0$ (when $m^P\geq m^*)$ or $0 < \alpha < B_0 - (p_2B_0/r)(m^*-m^P)$ (when $m^P < m^*$) the situation is more involved and $h^*>1$ does not automatically yield a high state prepayment region. To see this, assume $(1,\infty)$ is in the continuation region, so that $h>1$ implies $V^P(h) = \wt{C}_2 h^{-p_2} + m^P B_0 /r$ where $\wt{C}_2<0$. On $(1,\infty)$, the minimal difference between immediate payoff and $V^P(h)$ is
\begin{equation}\label{E:min_diff_f}
\begin{split}
&\inf_{h>1} \left(B_0 + \alpha(h-1) - \wt{C}_2 h^{-p_2}  - \frac{m^P B_0}{r}\right) = (1+p_2)\left(\frac{\alpha}{p_2}\right)^{\frac{p_2}{1+p_2}}(-\wt{C}_2)^{\frac{1}{1+p_2}} - \frac{B_0(m^P-r)+\alpha r}{r}.
\end{split}
\end{equation}
This quantity must be non-negative, which imposes the restriction
\begin{equation}\label{E:aprm_wtC2_cond}
    -\wt{C}_2 \geq \left(\frac{(m^P-r)B_0 + \alpha r}{(1+p_2)r}\right)^{1+p_2}\left(\frac{p_2}{\alpha}\right)^{p_2}.
\end{equation}
However, as $\wt{C}_2$ is determined by value matching and smooth pasting at $h=1$, it depends on what is taking place for low house states, and will change (see equations \eqref{E:m_leq_delta_beta} and \eqref{E:m_geq_delta_beta} below respectively) if $m^P\leq \delta$ or  $\delta < m^P$.   But, as we show in Remark \ref{R:alpha_star_new} below, for both $m^P\leq \delta$ and $\delta < m^P < m^*$ there is a unique $\alpha^* \in (0, p_2B_0(m^P/r-1) \subset (0,B_0)$ such that \eqref{E:aprm_wtC2_cond} holds if an only if $\alpha \geq \alpha^*$, and hence a high state prepayment region emerges only when $\alpha < \alpha^*$.  Lastly, when $m^P\geq m^*$, $\wt{C}_2$ takes the same form as \eqref{E:m_geq_delta_beta}, but now the unique  $\alpha^*  \in (0, p_2B_0(m^P/r-1)$ enforcing equality in \eqref{E:aprm_wtC2_cond} lies above $B_0$, violating Assumption \ref{A:aprm_alpha}. As such, for $0 < \alpha< B_0$ there is a high state prepayment region.

\end{rem}

Motivated by the above, we split the APRM into three cases: when $0 < m^P \leq \delta$;  when $\delta < m^P < m^*$; and when $m^P\geq m^*$. Accounting for turnover we still have
\begin{equation}\label{E:mstar_def}
m^* \dfn \frac{p_1\delta}{p_1-1}.
\end{equation}
but from \eqref{E:p1_p2} we see that $p_1$ depends on $\lambda$. At most, there will be three optimal prepayment boundaries. We will denote the one below $h=1$ by $h_1$, and two boundaries above $h=1$ by $h_2$ and $h_3$. Of course, even though we use the same notation for these boundaries, in general they do not coincide across different cases. 


\begin{prop}\label{P:APRM_P1} Assume $m^P\leq \delta$. Then there is $0 < \alpha^* < B_0$ such that 

(i) When $\alpha < \alpha^*$ the action regions and value function are
\begin{small}
\begin{equation*}
\begin{tabular}{c||c|c|c|c}
    $h$ & $< 1$ & $\in [1,h_2)$ & $\in [h_2,h_3]$ & $> h_2$ \\
    \hline
    Action & Continue & Continue & Prepay & Continue \\
    \hline
    $V^P(h)$ & $C_1h^{p_1} +V^{P}_{par}(h)$ & $\wt{C}_1h^{p_1} + \wt{C}_2h^{-p_2} + V^{P}_{par}(h)$ & $B_0 + \alpha(h-1)$ & $\check{C}_2 h^{-p_2} + V^{P}_{par}(h)$\\
\end{tabular}
\end{equation*}
\end{small}
where the constants $C_1, \wt{C}_1,\wt{C}_2,\check{C}_2$ are all negative and $h_2,h_3$ are optimal prepayment boundaries. 

(ii) When $\alpha\geq \alpha^*$ the action regions and value function are 
\begin{small}
\begin{equation*}
\begin{tabular}{c||c|c}
    $h$ & $< 1$ &  $> 1$ \\
    \hline
    Action & Continue & Continue \\
    \hline
    $V^P(h)$ & $K_1 h^{p_1}  + V^{P}_{par}(h) $ & $\wt{K}_2h^{-p_2} + V^{P}_{par}(h) $\\
\end{tabular}
\end{equation*}
\end{small}
where the constants $K_1,\wt{K}_2$ are all negative.
\end{prop}

\begin{figure}[t]
\includegraphics[height=5cm,width=7cm]{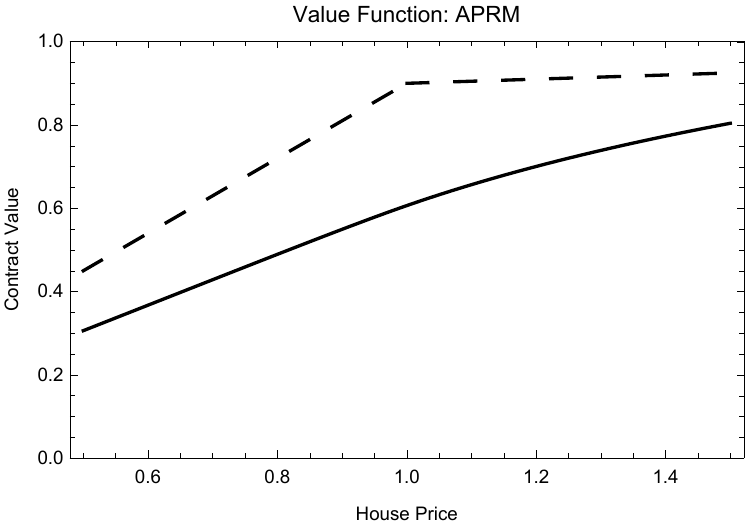}\quad
\includegraphics[height=5cm,width=7cm]{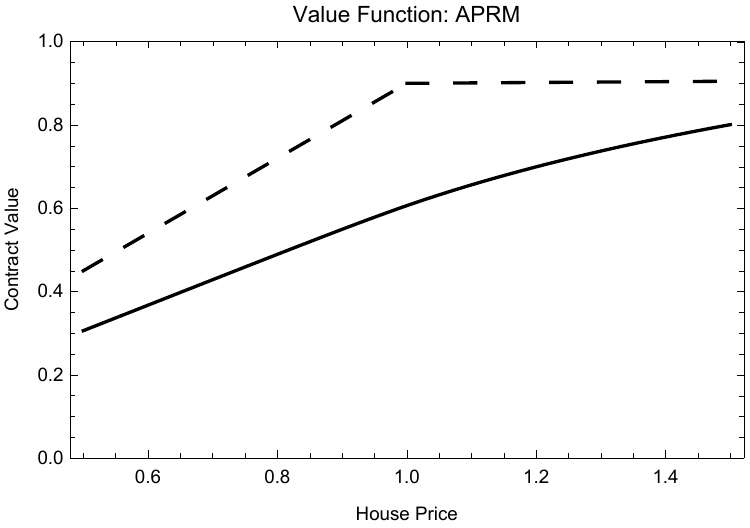} \\

\caption{The APRM value function $V^P(h)$ (solid) versus the payoff function $B_0\min(1,h) + \alpha(h-1)^+$ (dashed) when $m^P < \delta = 7\%$. All other parameters are in Table \ref{tab:parameters} with $m^P = m_0$ therein. Left panel: $\alpha= 5\% <\alpha^*=10.1\%$, with optimal boundaries $h_2=2.39$ and $h_3=7.60$. Right panel: $\alpha=1\%$ with boundaries $h_2 = 2.13$ and $h_3 = 37.94$.}
\label{F:value_aprm1}

\end{figure}

\begin{prop}\label{P:APRM_P2}
Assume $\delta < m^P < m^*$. Then there is $0 < \alpha^* < B_0$  such that 

For (i) $\alpha < \alpha^*$ the action regions and value functions are 
\begin{small}
\begin{equation*}
\begin{tabular}{c||c|c|c|c|c}
    $h$ & $\leq h_1$ & $\in (h_1,1]$ & $\in [1,h_2)$ & $\in [h_2,h_3]$ & $> h_3$ \\
    \hline
    Action & Prepay & Continue & Continue & Prepay & Continue \\
    \hline
    $V^P(h)$ & $B_0h$ & $C_1 h^{p_1} + C_2 h^{-p_2} +V^{P}_{par}(h)$ & $\wt{C}_1h^{p_1} + \wt{C}_2 h^{-p_2} + V^{P}_{par}(h)$ & $B_0 + \alpha(h-1)$ & $\check{C}_2 h^{-p_2} + V^{P}_{par}(h)$\\
\end{tabular}
\end{equation*}
\end{small}
where the constants $C_1,C_2,\wt{C}_1,\wt{C}_2,\check{C}_2$ are all negative, and $h_1,h_2,h_3$ are optimal prepayment boundaries. 

For (ii)  $\alpha\geq \alpha^*$ the action regions and value function are
\begin{small}
\begin{equation*}
\begin{tabular}{c||c|c|c}
    $h$ & $\leq h_1$ & $\in (h_1,1]$ & $> 1$ \\
    \hline
    Action & Prepay & Continue & Continue \\
    \hline
    $V^P(h)$ & $B_0h$ & $K_1 h^{p_1} + K_2 h^{-p_2} + V^{P}_{par}(h)$ & $\wt{K}_2 h^{-p_2} + V^{P}_{par}(h)$ \\
\end{tabular}
\end{equation*}
\end{small}
where the constants $K_1,K_2,\wt{K}_2$ are all negative, and $h_1$ is the optimal prepayment boundary.

\end{prop}

\begin{figure}[t]
\includegraphics[height=5cm,width=7cm]{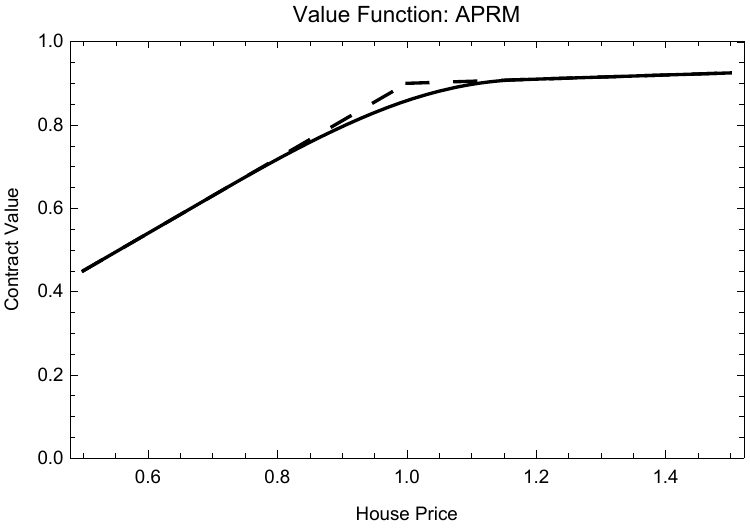}\quad
\includegraphics[height=5cm,width=7cm]{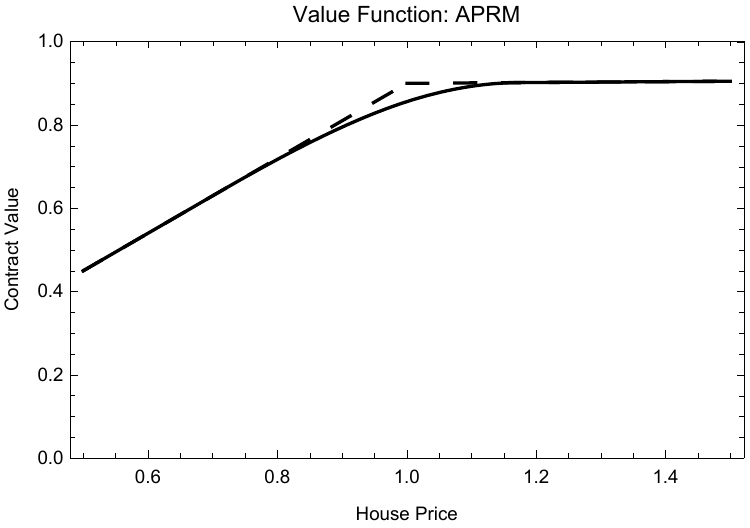} \\

\caption{ The APRM value function $V^P(h)$ (solid line) versus the payoff function $B_0\min(1,h) + \alpha(h-1)^+$ (dashed line) when $3\% = \delta < m^P < m^* = 5.08\%$. Parameters are in Table \ref{tab:parameters} with $m^P = m_0$ therein. Left panel: $\alpha=5\% < \alpha^* = 76\%$, with optimal boundaries $h_1=0.72$, $h_2=1.17$ and $h_3=16.39$. Right panel: $\alpha=1\%$ with boundaries $h_1=0.71$, $h_2 = 1.17$ and $h_3 = 81.81$.}
\label{F:value_aprm2}
\end{figure}

We conclude with the $m^P\geq m^*$ case.  As mentioned above, this forces $\alpha^*\geq B_0$ which violates Assumption \ref{A:aprm_alpha}.  Therefore, there is always a high-state prepayment region.  This leads to the following 

\begin{prop}\label{P:APRM_P3} Assume $m^P\geq m^*$. Then, the action regions and value function are
\begin{small}
\begin{equation*}
\begin{tabular}{c||c|c|c|c|c}
    $h$ & $\leq h_1$ & $\in (h_1,1]$ & $\in [1,h_2)$ & $\in [h_2,h_3]$ & $> h_3$ \\
    \hline
    Action & Prepay & Continue & Continue & Prepay & Continue \\
    \hline
    $V^P(h)$ & $B_0h$ & $C_1 h^{p_1} + C_2 h^{-p_2} + V^{P}_{par}(h)$ & $\wt{C}_1 h^{p_1} + \wt{C}_2 h^{-p_2} + V^{P}_{par}(h)$ & $B_0 + \alpha(h-1)$ & $\check{C}_2 h^{-p_2} + V^{P}_{par}(h)$\\
\end{tabular}
\end{equation*}
\end{small}
where the constants $C_1,C_2,\wt{C}_1,\wt{C}_2, \check{C}_2$ are all negative, and $h_1,h_2,h_3$ are optimal prepayment boundaries.
\end{prop}



Lastly, as with the ABM recall from Remark \ref{R:ABM_APRM_no_default_value} that the APRM eliminates the default option value and hence $V^{NoDef,P} \equiv 0$.



\subsection{Foreclosure Costs and Endogenous Spreads}\label{SS:foreclosure}

To gain a clearer picture of ABM and APRM's effectiveness, we now account for foreclosure costs. Indeed, should the borrower default at time $\tau$, the bank typically receives far less than the home price $H_{\tau}$, due to both direct and indirect foreclosure costs which may be $30-40\%$ of the home value (\cite{campbell2011forced, andersson2014loss}), and the primary reason the ABM and APRM have been proposed as beneficial to the banks is that despite the lower payment rates and outstanding balances, they are competitive with the traditional FRM when one accounts for foreclosure costs.  

We assume that upon default of the FRM at time $\tau$, there is a fractional loss $\phi$ incurred by the bank, so that the bank receives $(1-\phi)H_{\tau}$.  The borrower does not account for $\phi$, so it will not affect the optimal default time\footnote{One can allow $\tau$ to depend on $\phi$ by changing the action payoff from $\min\bra{B_0,H_\tau}$ to $\min\bra{B_0,(1-\phi)H_\tau}$ and Proposition \ref{P:FRM_P} goes through with minor adjustments. Our perspective is that $\phi$ should not change the default time.}.   Proposition \ref{P:FRM_P} implies the optimal stopping time, for a given starting house price level $h$, is $\tau(h) = \inf\cbra{t\geq 0 \such H_t \leq h_1 \textrm{ or } H_t \geq h_2}$ with default at $h_1$ and prepayment at $h_2$. Therefore, the FRM has foreclosure-adjusted value
\begin{equation*}
\begin{split}
    V^{F}_{\phi}(h) &= \espalt{}{h}{\int_0^{\tau(h)}  e^{-(r+\lambda)u}\left( m^F B_0 + \lambda\left((1-\phi)H_u 1_{H_u\leq B_0} + B_0 1_{H_u > B_0}\right)\right)du}\\
    &\qquad + \espalt{}{h}{e^{-(r+\lambda)\tau(h)}\left((1-\phi)H_{\tau(h)} 1_{H_{\tau(h)}\leq B_0} + B_0 1_{H_{\tau(h)}>B_0}\right)}\\
    &= V^F(h) - \phi\lambda \espalt{}{h}{\int_0^{\tau(h)}  e^{-(r+\lambda)u}H_u 1_{H_u\leq B_0}du} - \phi \espalt{}{h}{e^{-(r+\lambda)\tau(h)} H_{\tau(h)}1_{H_{\tau(h)} \leq B_0}}.
\end{split}
\end{equation*}
If $h\geq h_2>B_0$ then $\tau(h) = 0$ (immediate prepayment) and $V^F_{\phi}(h) = B_0$. Similarly, if $h\leq h_1 < B_0$ then $\tau(h) = 0$ (immediate default) and $V^F_{\phi}(h) = (1-\phi)h$.  For $h_1 < h < h_2$,  we must compute
\begin{equation*}
         u_1(h) \dfn \espalt{}{h}{\int_0^{\tau(h)}  e^{-(r+\lambda)u}H_u 1_{H_u\leq B_0}du}, \qquad u_2(h) \dfn \espalt{}{h}{e^{-(r+\lambda)\tau(h)} H_{\tau(h)}1_{H_{\tau(h)} \leq B_0}}.
\end{equation*}
We start with $u_1$.  Here, calculations very similar to those which lead to Proposition \ref{P:FRM_P} show that 
\begin{equation*}
    u_1(h) = w_1(h) - \espalt{}{h}{e^{-(r+\lambda)\tau(h)}w_1(H_{\tau(h)})};\quad w_1(h) \dfn \begin{cases} -\frac{(1+p_2)B_0^{1-p_1}}{(p_1+p_2)(\lambda+\delta)} h^{p_1} + \frac{h}{\lambda + \delta} & h < B_0 \\ \frac{(p_1-1)B_0^{1+p_2}}{(p_1+p_2)(\lambda + \delta)} h^{-p_2} & h\geq B_0\end{cases}
\end{equation*}
Using that $\tau(h) = \tau_1(h) \wedge \tau_2(h)$, where $\tau_i(h)$ is the first hitting time to $h_i,i=1,2$, we may further simplify this to
\begin{equation*}
    u_1(h) = w_1(h) - w_1(h_1)\espalt{}{h}{e^{-(r+\lambda)\tau_1(h)}1_{\tau_1(h) < \tau_2(h)}} - w_1(h_2)\espalt{}{h}{e^{-(r+\lambda)\tau_2(h)}1_{\tau_1(h) > \tau_2(h)}}.
\end{equation*}
Similarly
\begin{equation*}
    u_2(h) = h_1\espalt{}{h}{e^{-(r+\lambda)\tau_1(h)}1_{\tau_1(h) < \tau_2(h)}}.
\end{equation*}
Direct calculation shows for $h_1 < h < h_2$
\begin{equation*}
    \begin{split}
   \espalt{}{h}{e^{-(r+\lambda)\tau_1(h)}1_{\tau_1(h) < \tau_2(h)}} &= \left(\frac{h_1}{h}\right)^{p_2}\frac{h_2^{p_1+p_2}-h^{p_1+p_2}}{h_2^{p_1+p_2}-h_1^{p_1+p_2}}\\ \espalt{}{h}{e^{-(r+\lambda)\tau_2(h)}1_{\tau_1(h) > \tau_2(h)}} &= \left(\frac{h_2}{h}\right)^{p_2}\frac{h^{p_1+p_2}-h_1^{p_1+p_2}}{h_2^{p_1+p_2}-h_1^{p_1+p_2}}
   \end{split}
\end{equation*}
so that in $(h_1,h_2)$ we have an explicit expression for
\begin{equation}\label{E:FRM_foreclosure}
    V^{F}_{\phi}(h) = V^{F}(h) - \lambda\phi u_1(h) - \phi u_2(h).
\end{equation}
We may thus identify the equivalent foreclosure proportional costs $\phi^A = \phi^A(h)$ and $\phi^P = \phi^P(h)$ which, for fixed contract rates $m^F,m^A,m^P$, equate the adjusted FRM value $V^{F}_{\phi}(h)$ with the respective ABM and APRM values $V^{A}(h)$ and  $V^{P}(h)$. This in turn will tell us how large foreclosure costs need to be before the proposed contracts outperform the FRM.

A second way of comparing the contracts' performance accounting for foreclosure costs is to identify endogenous mortgage rates.  Here, for a given foreclosure percentage cost $\phi$ and FRM contract rate $m^F$, one seeks rates $m^A$ and  $m^P$ for which all three contracts have the same value. More precisely, if we think of the contract value as a function of both the house price and mortgage rate, then we use \eqref{E:FRM_foreclosure} to seek $m^A(\phi)$ and $m^P(\phi)$ such that
\begin{equation*}
    V^{F}_{\phi}(h,m^F) = V^{A}(h,m^A(\phi)) = V^{P}(h,m^P(\phi))
\end{equation*}
and identify the endogenous spread (in bps) as
\begin{equation}\label{E:endog_spread}
\mathfrak{s}^A(\phi) \dfn 10,000 \times (m^A(\phi)-m^F);\qquad \mathfrak{s}^P(\phi) \dfn 10,000 \times (m^P(\phi) - m^F).
\end{equation}

\section{\bf Numerical Analysis}\label{S:numerics}

We now numerically compare the three contracts, using the parameters in Table \ref{tab:parameters}. 
    \begin{table}
        \centering
        \begin{tabular}{|c|c|c|c|}
             \hline
             $r = 0.1034\%$ & $m_0 = 3.31\%$ & $B_0 = 0.90$ & $\sigma = 11.25\%$\\
             \hline
             $\lambda = 4.54\%$ & $\delta = 4.50\%$ or $7.00\%$ & $\alpha = 1.00\%$ or $5.00\%$ & \\
             \hline
        \end{tabular}
        \caption{Parameter Values}\label{tab:parameters}
    \end{table}
Therein, $r$ is the $1$-month US Libor rate as of January 4, 2022; $m_0$  is the Mortgage Banker's Association $30$ jumbo fixed rate as of January 3, 2022; $\sigma\%$ comes from \cite{greenwald2018financial}; and for $\lambda$, we follow \cite{hu2011asset}, estimating the turnover rate as the ratio of existing home sales to existing home stock (obtained from the St Louis fed as of January 1, 2022). The values for $\delta$  combine the $2\%$ benefit rate and $5\%$ rental cost rate of \cite{mei2019improving}, using either all or half of the rental cost rate.  Lastly, for $\alpha$ we take either $5\%$ (suggested in \cite{mian2013state, mian2015house}), or $\alpha = 1\%$ to offer a low penalty comparison.

As summarized in the introduction,  our main findings are 

\begin{enumerate}[(1)]
\item The APRM contract value is insensitive to the capital gain sharing proportion $\alpha$ because, even for small $\alpha$, high state prepayment is virtually eliminated.  Therefore, it is difficult to allow for endogenous $\alpha$ as one cannot invert the contract value in $\alpha$.
\item For a given common contract rate,  the APRM has a lower value than the ABM, even ignoring the capital gain sharing feature,   because the APRM lowers payments once $H$ falls below $1$, rather than once $H$ falls below $B_0$.
\item Depending on the benefit rate $\delta$, for relatively low foreclosure costs, the ABM may be more valuable than the FRM in low house price states even at a common contract rate. Furthermore, for all $\delta$ the ABM has a lower equivalent foreclosure cost than the APRM.  
\item For observed foreclosure costs (e.g. $30\%-35\%$) the endogenous spread of the ABM is lower than that for the APRM, but both increase substantially with the utility rate $\delta$. However, for low utility rates, at observed foreclosure rates, the ABM actually has a negative endogenous spread.   
\end{enumerate}

\begin{rem}\label{R::marketing}
Item $(4)$ above implies the ABM could be effectively marketed, even when the benefit rate is high. Indeed, at almost the same mortgage rate as the traditional FRM,  the ABM will ensure the depositor is never underwater. As the benefit from remaining in the house is large relative to her mortgage rate, the borrower will not prepay the mortgage at low price values.  Thus, the borrower gets default protection, and the bank gets a fairly valued mortgage, with little to no possibility of receiving the house in low house price states. In view of the insensitivity of the APRM contract value to the sharing penalty $\alpha$, we conclude the ABM is more effective at preventing ruthless defaults; being palatable to the borrower; and not introducing unexpected prepayment behaviors. 

\end{rem}

We begin by investigating the sensitivity of the APRM contract value to $\alpha$. Figure \ref{F:srm_alpha_maps} shows the map $\alpha\to V^{P}(h;\alpha)$ for three different values of $h$ (recall that $H_0=1$ so the different values of $h$ correspond to $H_t$ at some future time $t>0$.). We see that  $\alpha$ has a minimal effect on the contract value, with $V^P$ only visibly increasing in $\alpha$ for the high house state.  This is because if $\alpha$ exceeds the threshold $\alpha^*$ from Propositions \ref{P:APRM_P1} and \ref{P:APRM_P2} (in this case, $\alpha^* = 10.1\%$ (high $\delta$) or $\alpha^* = 28.2\%$ (low $\delta$)), the borrower never prepays the APRM when $h>1$, and the contract value only depends on $\alpha$ through turnover, and not strategic default. However, even for $\alpha < \alpha^*$, the high state prepayment boundary $h_2$ for the APRM is large enough to ensure a minimal dependence of $V^P$ on $\alpha$.  

\begin{figure}[t]
\includegraphics[height=5cm,width=7cm]{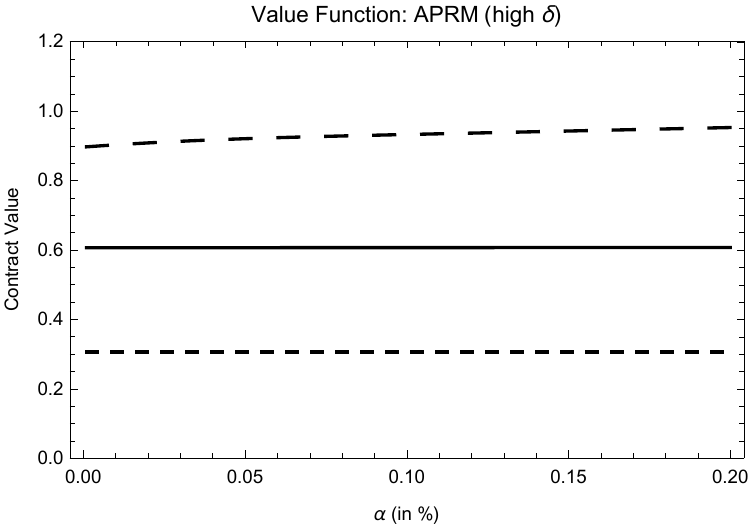}\qquad \includegraphics[height=5cm,width=7cm]{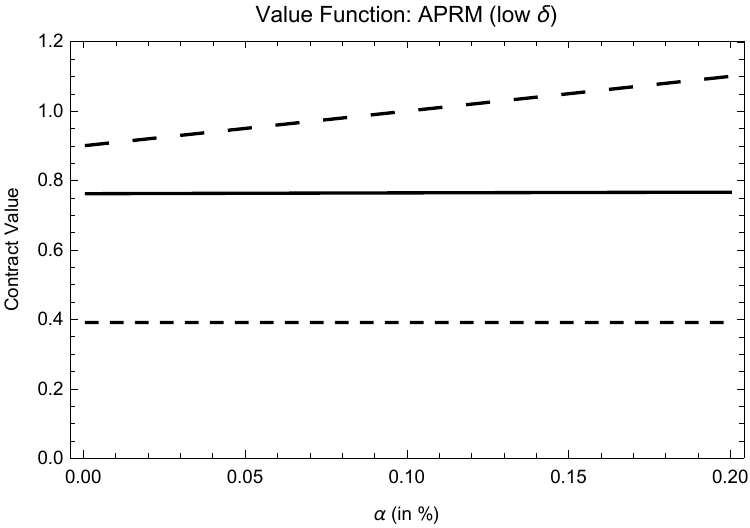}\\
\caption{The map $\alpha \to V^{P}(h;\alpha)$ for $h=0.5$ (thin dash), $h=1$ (solid) and $h=2$ (thick dash) with $\delta = 7\%$ (left) and  $\delta = 4.5\%$ (right). Parameters are in Table \ref{tab:parameters} with $m^P = m_0$ therein.}
\label{F:srm_alpha_maps}
\end{figure}

 We next identify equivalent foreclosure costs for a common contract rate $m_0=m^F=m^A=m^P$. Here, Figure \ref{F:foreclosure_costs} plots maps $H_t\to \phi^A(H_t)$ and $H_t\to \phi^P(H_t)$ at some future time $t>0$ for house prices below the initial $H_0=1$. First, we observe the equivalent foreclosure rates are relatively insensitive to the house price. Second, equivalent foreclosure costs are small. For example, when $\delta=4.5\%$, the ABM becomes more valuable than the FRM once foreclosure costs approach $15\%$ of the home value. The APRM fares worse, with equivalent costs in the $20\%-30\%$ range, but nevertheless both contracts outperform at the observed foreclosure costs of \cite{andersson2014loss, campbell2011forced}. For $\delta=7\%$, foreclosure costs must be higher before the ABM and APRM contracts are competitive. Here, the larger benefit the homeowner obtains through home-ownership induces her to remain in the house. As such, the lower payment rate dominates, reducing the ABM and APRM contract values.  Lastly, the flat regions correspond to where immediate default is the optimal FRM strategy. 


\begin{figure}[!t]
\includegraphics[height=5cm,width=7cm]{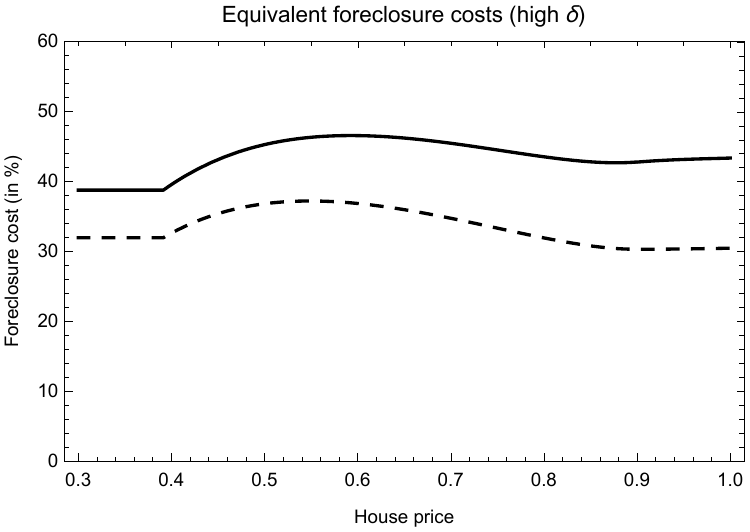}\qquad \includegraphics[height=5cm,width=7cm]{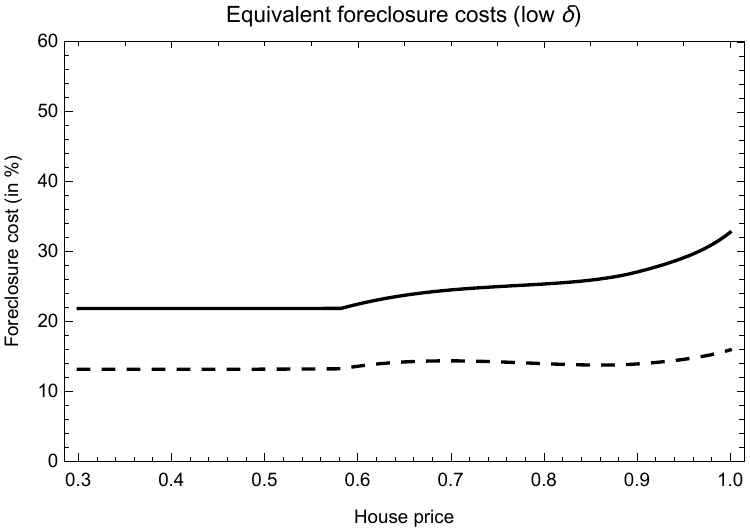}\\
\caption{Equivalent foreclosure costs (in $\%$) for ABM (dashed) and APRM ($\alpha=5\%$) (solid) as a function of the house price with $\delta = 7.0\%$ (left), and $\delta = 4.5\%$ (right). All other parameters are in Table \ref{tab:parameters} with $m^F=m^A=m^P = m_0$ therein.}
\label{F:foreclosure_costs}
\end{figure}

 While providing an interesting viewpoint for the ABM and APRM, identifying equivalent foreclosure costs for a common mortgage rate does not really equate the contract values, as foreclosure costs are both exogenous and fixed for a given locale (though varying widely across the United States).  Therefore, we turn to a second method for equating the contracts: identifying for a fixed foreclosure cost, the mortgage rate spread which equates ABM and APRM values to FRM value. Figure \ref{F:spreads} plots the maps $\phi \to \mathfrak{s}^A(\phi)$ and $\phi\to \mathfrak{s}^P(\phi)$ at $H_0=1$ for the spreads $\mathfrak{s}^A$,  $\mathfrak{s}^P$ from \eqref{E:endog_spread}. For example, in the case of $\delta=4.5\%$, if foreclosure costs are $15\%$ of the home value, then the ABM contract need only offer a spread of $3.6$ basis points before it has the same value as the FRM, while the APRM must offer a higher spread of $70.8$ basis points.  For costs above $16\%$ the ABM can actually offer a lower mortgage rate than the FRM.  For $\delta = 7\%$ the spreads dramatically  increase, with, e.g. at $25\%$ foreclosure costs, the ABM needing to offer $27.7$ basis points and the APRM 101.1 basis points. 

\begin{figure}[!t]
\includegraphics[height=5cm,width=7cm]{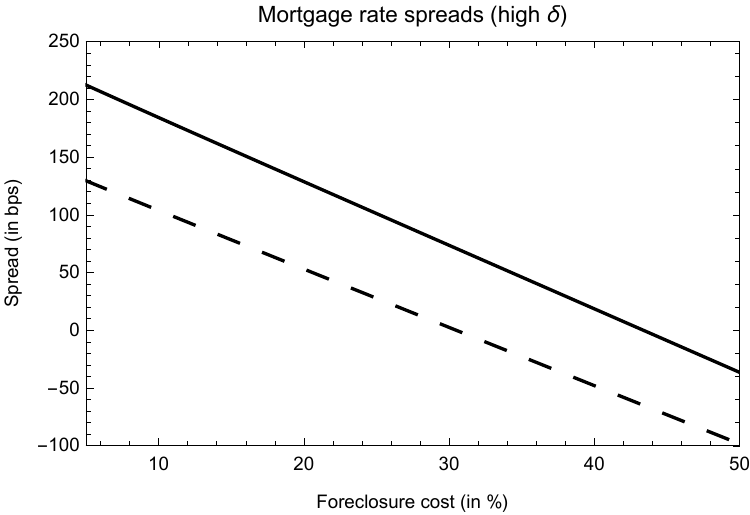}\qquad \includegraphics[height=5cm,width=7cm]{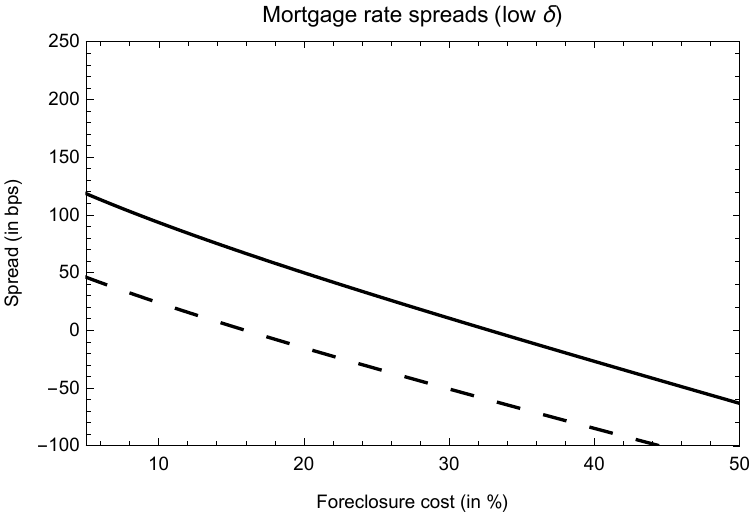}\\
\caption{Endogenous mortgage rate spreads (in basis points) at $H_0=1$, as a function of the foreclosure cost for ABM (dashed) and APRM ($\alpha=5\%$, solid) for $\delta = 7\%$ (left) and $\delta = 4.5\%$ (right). All other parameters are in Table \ref{tab:parameters}.}
\label{F:spreads}
\end{figure}

 We conclude this section by analyzing the value functions  and action boundaries.   In light of the discussion above, we set the ABM and APRM contract rates to ensure that at origination (when $H_0 = 1$) the contracts have the same value as the FRM (with contract rate $m_0 = 3.31\%$) assuming  $35\%$ foreclosure costs (c.f. \cite{hatcher2006foreclosure, campbell2011forced, andersson2014loss}). This gives the contract rates
\begin{equation}\label{E:adj_rates}
\begin{split}
    m^F &= 3.31\%;\qquad m^A(35\%) =  2.63\%;\qquad m^P(35\%) = 3.22\%\qquad\qquad (\delta = 4.5\%)\\
    m^F &= 3.31\%;\qquad m^A(35\%) =  3.08\%;\qquad m^P(35\%) = 3.77\%    \qquad\qquad (\delta = 7\%)
 \end{split}
\end{equation}
Figure \ref{F:vf_alpha05} plots the values of contracts as a function of the house price $H_t$ at some future time $t>0$. The ABM and APRM have the rates above, and the FRM accounts for the $35\%$ foreclosure cost. To interpret this plot, recall that we have normalized the index at $t=0$ so  that $H_0=1$. Here, a nice phenomena occurs: while we have determined $m^A,m^P$ to equate the value functions at $h=1$, the ABM and APRM value functions stay close to one another other across a wide range of $h$, especially in the low $\delta$ default region for the FRM where the ABM, APRM eliminate selective default and hence the $35\%$ loss due to the foreclosure process.


\begin{figure}[!t]
\includegraphics[height=5cm,width=7cm]{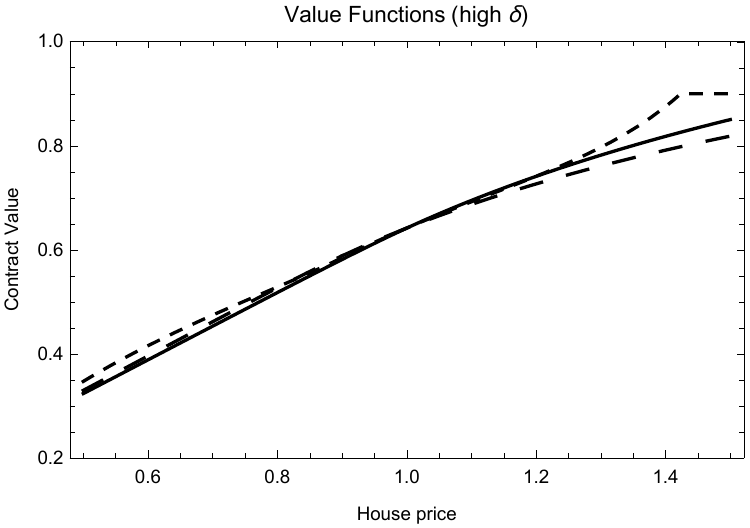} \qquad\includegraphics[height=5cm,width=7cm]{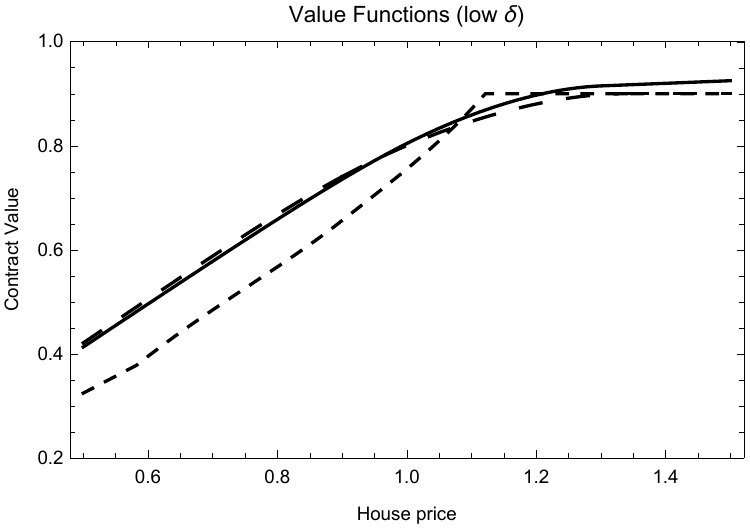}\\
\caption{Value functions for FRM (thin dash, adjusted for foreclosure costs), ABM (thick dash), and APRM ($\alpha=5\%$, solid)  as a function of the  house price $H_t$ at some time $t>0$ for $\delta = 7.0\%$ (left) and $\delta = 4.5\%$ (right). The foreclosure cost is $\phi = 35\%$ and mortgage rates are given in \eqref{E:adj_rates}. All other parameters are given in Table \ref{tab:parameters}.}
\label{F:vf_alpha05}
\end{figure}

 Lastly, we investigate the action boundaries. Here, we vary the FRM rate $m^F$ over $(r,\ol{m}^F)$, where $\ol{m}^F$ is the largest $m^F$ for which immediate prepayment is not the optimal policy at initiation ($t=0, H_0=1$).  The maximum rate $\ol{m}^F$ takes the values $4.86\%$ and $6.16\%$ for $\delta=4.5\%$ and $\delta = 7\%$, respectively.  Then, for each $m^F$ in this range, we identify the associated $m^A$ and $m^p$ which equate the contract values assuming foreclosure costs of $\phi = 35\%$ and we plot the optimal stopping boundaries for all three contracts at their respective rates. Note for the FRM that the lower boundary is a default boundary, while for the ABM and APRM the lower boundary is a low-state prepayment boundary. Lastly, note in Figure \ref{F:bdy_alpha05} there is a large state prepayment boundary $h_3$ for the APRM when $m^F$ high enough to ensure $\alpha = 5\% < \alpha^*$.   However, it does not appear in the pictures because it is on the order of $6 B_0-15B_0$ and for all practical purposes, is irrelevant. 


\begin{figure}[!t]
\includegraphics[height=5cm,width=7cm]{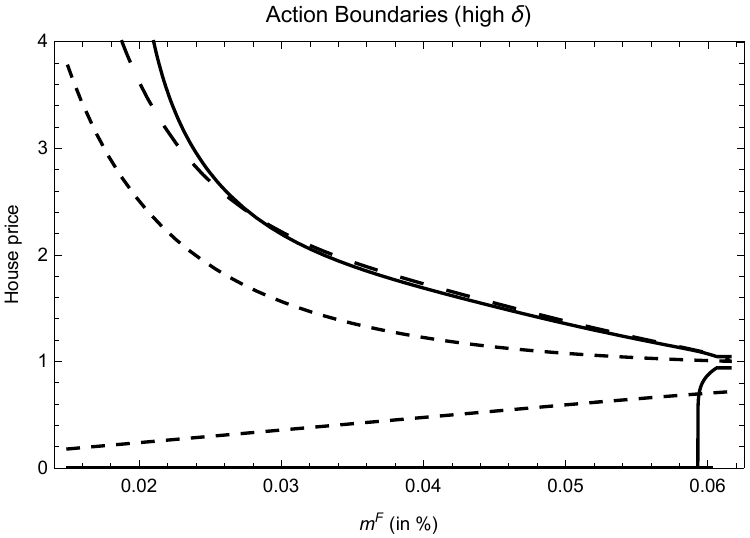} \qquad \includegraphics[height=5cm,width=7cm]{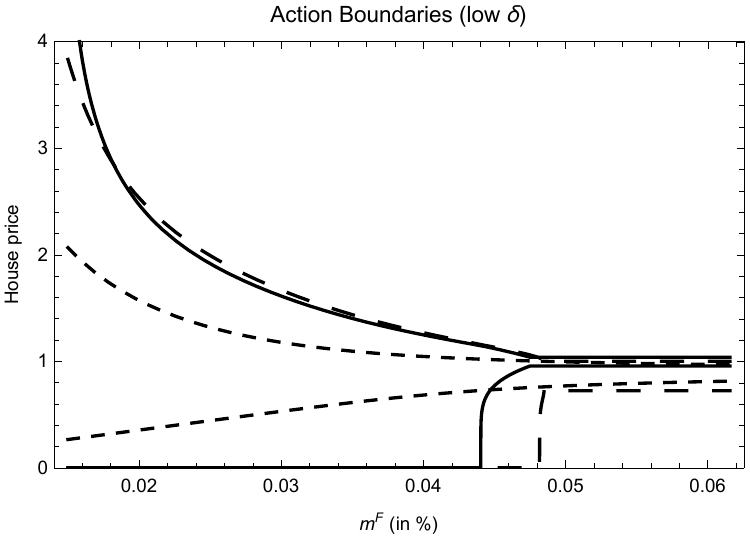}\\
\caption{Action boundaries for FRM (thin dash), ABM (thick dash), APRM  ($\alpha=5\%$, solid) as a function of the FRM mortgage rate $m^F$ for  $\delta = 7\%$ (left) and $\delta = 4.5\%$ (right).  Contract rates $m^A$ and $m^F$ are set to equate contract values at $H_0=1$ and $35\%$ foreclosure costs for a given FRM rate $m^F$. For the FRM, the upper curve is the prepayment boundary and lower curve is the default boundary.  For the ABM and APRM both curves are prepayment boundaries. Continuation in the middle. Parameters are given in Table \ref{tab:parameters}.}
\label{F:bdy_alpha05}
\end{figure}

\section{Finite Horizon Comparison}\label{S:finite_horizon}

In this section, our goal is to (briefly) justify the assumptions of an infinite horizon and constant interest rates, by showing the default boundary is insensitive to the interest rate and remaining maturity, at least near the beginning of the mortgage contract term. We focus on the finite horizon FRM, and to highlight the role of strategic behavior, assume no mortgage turnover. In this setting, recall the house price dynamics in \eqref{E:H_dynamics}, and assume that rather than being constant, the interest rate $r$ follows a CIR process with dynamics
\begin{equation*}
    dr_t = \kappa(\theta-r_t)dt + \xi\sqrt{r_t}dB_t;\qquad d\langle W,B\rangle_t = \rho dt
\end{equation*}
where $|\rho| < 1$ and $\kappa,\theta,\xi>0$ are constants with  $\kappa\theta \geq \xi^2/2$, to ensure $r$ stays strictly positive. Fix a time $t < T$, where $T$ is the mortgage maturity. Provided no default or prepayment by $t$, the FRM value function takes the form $V^F(t,r_t,H_t)$ where (c.f. \eqref{E:fin_horizon_vf})
\begin{equation*}
    V^F(t,r,h) = \underset{\tau\in[t,T]}{\operatorname{inf}}\espalt{}{t,r,h}{\int_t^\tau e^{-\int_t^u r_v dv} c^Fdu + e^{-\int_t^\tau r_v dv}\min\bra{H_{\tau},B^F_{\tau}}}
\end{equation*}
where $B^F$ and $c^F$ are given in \eqref{E:frm_balance} and \eqref{E:frm_coupon} respectively. As $B^F_T = 0$, the associated free boundary PDE is
\begin{equation*}
\begin{split}
\min\bra{V_t + \Lcal V - r V + c^F, \min\bra{h,B^F_t}-V}(t,r,h) = 0;&\qquad 0< t < T, r,h>0\\
V(T,r,h) =0;&\qquad r,h > 0,
\end{split}
\end{equation*}
where $\Lcal$ is the second order operator associated to $(r,H)$, given by
\begin{equation*}
    \Lcal V = \frac{1}{2}\xi^2 r V_{rr} + \rho\xi\sigma\sqrt{r}h V_{rh} + \frac{1}{2}\sigma^2 h^2 V_{hh} + \kappa(\theta-r) V_r + (r-\delta)h V_h,
\end{equation*}
and we solve the PDE backwards using finite differences.  Our interest lies in discovering how the action regions vary jointly with the interest rate and house price.  To this end, Figure \ref{F:cir_finite_horizon} shows the action regions in years $0$ and $5$ of a $30$-year mortgage.  Here, we see the default region (shaded light grey) is, to the naked eye, insensitive to the interest rate (the x-axis plot range spans the $3\%-97\%$ quantiles for the invariant distribution of $r$) and does not drastically change with the remaining horizon. As such, at a very broad level, we expect the conclusions obtained in the perpetual, constant interest rate case transfer over.

\begin{figure}[!t]
\includegraphics[height=5cm,width=7cm]{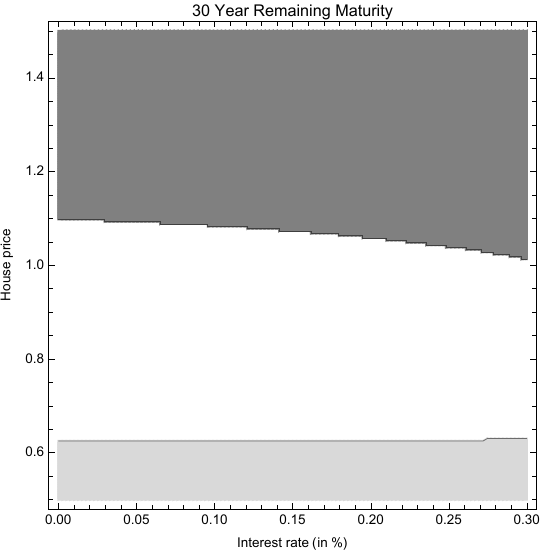} \qquad \includegraphics[height=5cm,width=7cm]{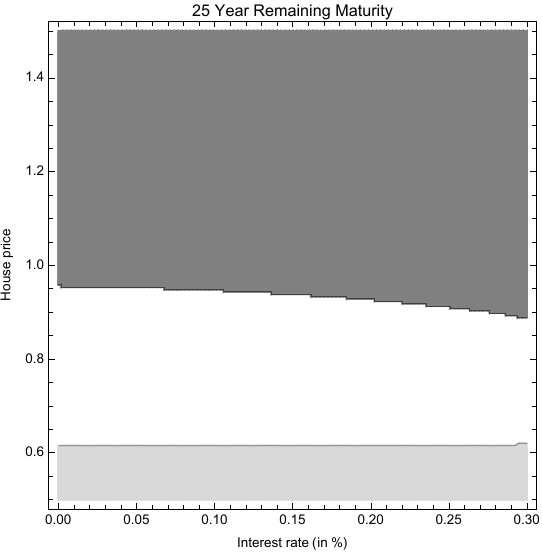}\\
\caption{Action boundaries for FRM as a function of the interest rate (x-axis) and house price (y-axis) in the finite horizon, stochastic interest rate setting.  The dark grey region corresponds to prepayment, and the light grey region to default. The left plot is at time $t=0$ and the right plot at time $t=5$ for a $30$ year mortgage contract.  CIR parameters are $\theta = .1034\%$, $\kappa = 1.5$ and $\xi = .04$. The correlation is $\rho = -0.5$. The house price utility is $\delta = 7\%$. Other parameters are given in Table \ref{tab:parameters}.}
\label{F:cir_finite_horizon}
\end{figure}

\section{\bf Conclusion}\label{S:conclusion}

In this article, we consider recently proposed mortgage contracts that aim to eliminate the negative effects due to underwater default.   In a continuous time model with  diffusive house prices, constant interest rates, and in the limit of a long maturity, we provide explicit solutions for the contracts' values assuming a worst-case approach to valuation.  We show that capital sharing features are ineffective, and to the extent that they make as few adjustments as possible while still ensuring the borrower is never underwater, the contracts become competitive with the traditional fixed rate mortgage at observed foreclosure costs and relatively low spreads.  While low-state prepayments are theoretically possible, provided the borrower receives sufficient utility living in the house (alternatively, that the costs associated with a default, such as rents and downward credit adjustments, are sufficiently high), low-state prepayment does not occur.  For future study, we aim to extend theoretical results to a finite horizon and stochastic interest rates, allow for a richer set of borrower strategic behaviors, and to incorporate basis risk between the observed local house price index value and the (partially) observed ``true'' house value. 

\bibliographystyle{siam}
\bibliography{master}

\newpage
\appendix

\section{\bf Proofs from Section \ref{S:perpetual}}\label{AS:S4_proofs}

\begin{proof}[Proof of Proposition \ref{P:verification}]
Define the local martingale $M_\cdot \dfn \int_0^\cdot e^{-(r+\lambda)u} H_u V_h(H_u)dW_u$.  First, we claim \eqref{E:dotV_integ} implies $M$ is a martingale such that for all $\filt^W$ stopping times $\tau$ we have  $\espalt{h}{}{M_{\tau}}{} = 0$.  Indeed, the Burkholder-Davis-Gundy inequalities imply $\espalt{h}{}{\sup_{t\geq 0} |M_t|} < \infty$ and hence $M$ is closable by the martingale $t\to \condespalt{h}{\sup_{u\geq 0}|M_u|}{t}$. Therefore, the result follows by optional sampling, which now is applicable for unbounded stopping times which may even be infinite with positive probability. Next, Ito's formula (valid for for $V$ under our assumptions: see \cite{MR2167640}) implies for any $\tau$ and integer $n$ that
\begin{equation*}
\begin{split}
e^{-(r+\lambda)\tau\wedge n}V(H_{\tau\wedge n}) &= V(h) + M_{\tau\wedge n} + \int_0^{\tau\wedge n} e^{-(r+\lambda)u}\left(\Lcal-(r+\lambda)\right)V(H_u)du\\
&= V(h) + M_{\tau\wedge n} - \int_0^{\tau\wedge n} e^{-(r+\lambda)u} c(H_u)du\\
&\qquad\qquad + \int_0^{\tau\wedge n} e^{-(r+\lambda)u}\left(c + (\Lcal-(r+\lambda))V\right)(H_u)1_{H_u\in E}du.
\end{split}
\end{equation*}
Under our assumptions, $e^{-(r+\lambda)\tau\wedge n} V(H_{\tau\wedge n})$, $M_{\tau\wedge n}$ and $\int_0^{\tau\wedge n} e^{-(r+\lambda)u}c(H_u)du$ have limits as $n\uparrow\infty$, even on the (potentially empty) set $\cbra{\tau=\infty}$.  This implies we make take $n\uparrow\infty$ above to obtain
\begin{equation*}
\begin{split}
e^{-(r+\lambda)\tau}V(H_{\tau}) &= V(h) + M_{\tau} - \int_0^{\tau} e^{-(r+\lambda)u} c(H_u)du + \int_0^{\tau} e^{-(r+\lambda)u}\left(c + (\Lcal-r)V\right)(H_u)1_{H_u\in E}du.
\end{split}
\end{equation*}
Thus, since $V$ solves \eqref{E:fbp} and $e^{-r\tau} f(H_\tau)$ is well defined, we conclude
\begin{equation*}
\begin{split}
e^{-(r+\lambda)\tau}f(H_{\tau}) + \int_0^{\tau} e^{-(r+\lambda)u} c(H_u)du &\geq e^{-(r+\lambda)\tau}V(H_{\tau}) + \int_0^{\tau} e^{-(r+\lambda)u} c(H_u)du \geq V(h) + M_{\tau}
\end{split}
\end{equation*}
with equality at the  candidate optimal $\tau^*$. The result readily follows by taking expectations.
\end{proof}


\subsection{FRM Proofs}

\begin{proof}[Proof of Proposition \ref{P:FRM_P}]

 We postulate that there are two prepayment boundaries 
$h_1,h_2$ such that $h_1<B_0<h_2$. The value matching and smooth pasting at $h_2$ imply
\begin{equation*}
\begin{split}
   &\wt C_1 h_2^{p_1} + \wt C_2 h_2^{-p_2}+\frac{m+\lambda}{r+\lambda}B_0=B_0\\
  & \wt C_1 p_1 h_2^{p_1-1} - \wt C_2 p_2 h_2^{-p+ 2-1}=0
\end{split}
\end{equation*}
so that
\begin{equation}\label{abm:wtC12n}
\begin{split}
&\wt{C}_{1}= h_{2}^{-p_{1}}B_{0}\frac{p_{2}}{p_1+p_2} \frac{r-m}{r+\lambda} <0\\
&\wt{C}_{2}=h_{2}^{p_{2}}B_{0}\frac{p_{1}}{p_1+p_2} \frac{r-m}{r+\lambda}<0
\end{split}
\end{equation}
as $m>r$.  Now the value matching and smooth pasting at $h_1$ give
\begin{equation*}
\begin{split}
   & C_1 h_1^{p_1} + C_2 h_1^{-p_2}+\frac{m}{r+\lambda}B_0 + \frac{\lambda}{\lambda+\delta}h_1=h_1\\
  &C_1 p_1 h_1^{p_1-1} - C_2 p_2 h_1^{-p_2-1}+\frac{\lambda}{\lambda+\delta}=1
\end{split}
\end{equation*}
so that
\begin{equation}\label{abm:C12n}
\begin{split}
&C_{1}=-\frac{h_1^{-p_1}}{p_{1}+p_{2}}\left(\frac{p_2mB_0}{r+\lambda}-(1+p_2)h_1\frac{\delta}{\delta+\lambda}\right)<0\\
&C_{2}=-\frac{h_1^{p_2}}{p_{1}+p_{2}}\left(\frac{p_1mB_0}{r+\lambda}-(p_1-1)h_1\frac{\delta}{\delta+\lambda}\right)<0.
\end{split}
\end{equation}
 Finally, imposing the  value matching and smooth pasting conditions at $B_0$ we get
\begin{equation*}
\begin{split}
&\widetilde{C}_1 B_{0}^{p_{1}}+\widetilde{C}_{2} B_{0}^{-p_{2}}+\frac{m+\lambda}{r+\lambda}B_0=C_{1} B_{0}^{p_{1}}+C_{2} B_{0}^{-p_{2}}+\frac{m}{r+\lambda}B_0 + \frac{\lambda}{\lambda+\delta}B_0\\
&\widetilde{C}_1 p_1 B_{0}^{p_{1}-1}-\widetilde{C}_{2}p_2 B_{0}^{-p_{2}-1}=C_{1} p_1B_{0}^{p_{1}-1}-C_{2}p_2 B_{0}^{-p_{2}-1}+\frac{\lambda}{\lambda+\delta}
\end{split}
\end{equation*}
that gives
\begin{equation*}
\begin{split}
&\widetilde{C}_{1}=C_1-\frac{p_2\frac{\lambda}{r+\lambda}-(p_2+1)\frac{\lambda}{\lambda+\delta}}{B_0^{p_1-1}(p_1+p_2)}=C_1-\frac{\lambda p_2}{(r+\lambda)(1-p_1)B_0^{p_1-1}(p_1+p_2)}\\
&\widetilde{C}_{2}=C_2-\frac{p_1\frac{\lambda}{r+\lambda}-(p_1-1)\frac{\lambda}{\lambda+\delta}}{B_0^{-p_2-1}(p_1+p_2)}=C_2-\frac{\lambda p_1}{(r+\lambda)(1+p_2)B_0^{-p_2-1}(p_1+p_2)}.
\end{split}
\end{equation*}

Now we set $h_1 = B_0 x$ and $h_2 = B_0 y$. Then we insert ? and ? into two equations above and obtain the system of two algebraic equations for $x$ and $y$
\begin{equation*}
\begin{split}
& \left(\frac{y}{x}\right)^{-p_{1}}=-x\frac{\delta p_1 }{(p_1-1)(m-r)}+\frac{m}{m-r}-x^{p_1}\frac{\lambda}{(p_1-1)(m-r)}\\
&\left(\frac{y}{x}\right)^{p_{2}}=-x\frac{\delta p_2 }{(p_2+1)(m-r)}+\frac{m}{m-r}+x^{-p_2}\frac{\lambda}{(p_2+1)(m-r)}.
\end{split}
\end{equation*}
The goal is to show that this system has unique solution $(x,y)$ with $0<x<1$ and $y>1$. Expressing $y$ in terms of $x$ in both equations yields the equation for $x$
\begin{equation*}
\begin{split}
1 = \chi(x) \dfn& \left(-x\frac{\delta p_1 }{(p_1-1)(m-r)}+\frac{m}{m-r}-x^{p_1}\frac{\lambda}{(p_1-1)(m-r)}\right)^{p_2}\\
&\times\left(-x\frac{\delta p_2 }{(p_2+1)(m-r)}+\frac{m}{m-r}+x^{-p_2}\frac{\lambda}{(p_2+1)(m-r)}\right)^{p_1}.
\end{split}
\end{equation*}
It is clear that $\chi$ is strictly decreasing, and $\chi(0) = \left(\frac{m}{m-r}\right)^{p_1+p_2} > 1$ if $\lambda=0$ and $\chi(0)=+\infty$ if $\lambda>0$. Next, it is easy to see that 
\begin{equation*}
    -x\frac{\delta p_1 }{p_1-1}+m-x^{p_1}\frac{\lambda}{p_1-1}\le -x\frac{\delta p_2 }{p_2+1}+m+x^{-p_2}\frac{\lambda}{p_2+1}
\end{equation*}
for $0<x<1$. Now let us consider two cases: 1) $m\le \frac{\delta p_1 +\lambda}{p_1-1}$; 2) $m>\frac{\delta p_1 +\lambda}{p_1-1}$. In the former case, $\chi(1)<0$ and hence there exists unique $\hat{x}\in(0,1)$ so that $\chi(\hat{x})=1$.  In the latter case, $\chi(1)>0$. Tedious algebra shows that that $\chi(1)<1$, i.e., $1=\chi(x)$ has unique solution on $(0,1)$.

Given that $C_1,C_2,\widetilde{C}_{1},\widetilde{C}_{2}<0$, the candidate solution satisfies all  the hypotheses of Proposition \ref{P:verification}, i.e., finishing the result.
 
\end{proof}


\begin{proof}[Proof of Proposition \ref{P:FRM_P_option}]

We again drop the ``F'' superscripts. Recall the definition of $V^{NoDef}$
\begin{equation*}
V^{NoDef}(h) = \inf_{\tau \geq 0} \espalt{h}{}{\int_0^\tau e^{-(r+\lambda)u}(m+\lambda)B_0du + e^{-(r+\lambda)\tau} B_0}.
\end{equation*}
Note that
\begin{equation*}
\partial_t \left(\int_0^t e^{-(r+\lambda) u} (m+\lambda) B_0 du + e^{-(r+\lambda)t}B_0\right)  = B_0(m-r)e^{-(r+\lambda)t}.
\end{equation*}
Thus, the optimal policy is to immediately prepay as $m>r$. 

\end{proof}


\subsection{ABM Proofs}

\begin{rem}\label{R:turnover_ABM} Upon inspection of \eqref{E:abstract_v}, \eqref{E:perpetual_flow_payoff}, \eqref{E:H_op} we see that for the ABM, the effect of turnover is to increase the interest rate, mortgage rate and utility rate by $\lambda$. Indeed, the value functions including turnover coincide with those excluding turnover, but shifting $r \to r+\lambda$, $\delta \to \delta +\lambda$ and $m^A \to m^A +\lambda$. This perspective is used in the proofs below, where we assume $\lambda = 0$ and $r,\delta,m^A$ have been adjusted.
\end{rem}

\begin{proof}[Proof of Proposition \ref{P:ABM_P}]

Throughout, we drop all ``A'' superscripts. First, for $m\leq \delta$ the value matching and smooth pasting conditions at $B_0$ give
\begin{equation}\label{E:abm_1}
\begin{split}
 C_1 B_{0}^{p_{1}}+m B_{0} / \delta&=\widetilde{C}_1 B_{0}^{p_{1}}+\widetilde{C}_2 B_{0}^{-p_{2}}+m B_{0}/r\\
  C_1 p_1 B_{0}^{p_{1}}+mB_0/ \delta&=\widetilde{C}_1 p_1 B_{0}^{p_{1}}-\widetilde{C}_2 p_2 B_{0}^{-p_{2}}.
\end{split}
\end{equation}
By multiplying the first equation of \eqref{E:abm_1} by $p_1$ and subtracting from it the second equation of \eqref{E:abm_1}  we obtain 
\begin{equation*}
 \widetilde{C}_2=-m B_{0}^{1+p_{2}}\left(\frac{p_{1}}{r}-\frac{p_{1}-1}{\delta}\right)/\left(p_{1}+p_{2}\right)
\end{equation*}
and clearly the right-hand side is negative (recall that $r<m\le \delta$). Hence, $\widetilde{C}_2<0$. 
Now we multiply the first equation \eqref{E:abm_1} by $p_2$ and add to it the second equation of \eqref{E:abm_1}   to get the relationship between $C_1$ and $\widetilde{C}_1$
\begin{equation}\label{E:abm_2}
   \left(p_{1}+p_{2}\right) \widetilde{C}_1=
     \left(p_{1}+p_{2}\right) C_1+m B_{0}^{1-p_{1}}\left(\frac{1+p_{2}}{\delta}-\frac{p_2}{r}\right).
\end{equation}
Now the value matching and smooth pasting conditions at $h_2$ give
\begin{equation}\label{E:abm_3} 
\begin{split}
\widetilde{C}_1 h_2^{p_{1}}+\widetilde{C}_2 h_2^{-p_{2}}+m B_0/r &=B_0\\
\widetilde{C}_1 p_1 h_2^{p_{1}}-\widetilde{C}_2 p_2 h_2^{-p_{2}} &=0.
\end{split}
\end{equation}
Let us apply multiply the first equation \eqref{E:abm_3} by $p_1$ and subtract from it the second equation \eqref{E:abm_3}
\begin{equation*}
   \left(p_{1}+p_{2}\right) \widetilde{C}_2 h_2^{-p_{2}}=
-m B_0p_1\left(\frac{1}{r}-\frac{1}{m}\right)
\end{equation*}
and as we know already $\widetilde{C}_2$ we deduce the optimal prepayment threshold
\begin{equation*}
    h_2 = B_0\left(\frac{1/r-(1-1/p_1)/\delta}
    {1/r-1/m}\right)^{1/p_2}.
\end{equation*}
Using that $m\le \delta$, we get $h_2>B_0$ as needed. Now by multiplying the first equation \eqref{E:abm_3} by $p_2$ and adding to it the second equation of \eqref{E:abm_3}, we obtain
\begin{equation*}
    \widetilde{C}_1=\frac{p_2}{p_1+p_2} mB_0 \left(\frac{1}{r}-\frac{1}{m}\right)
   h_2^{-p_1}
\end{equation*}
which is negative as $m>r$. Having found $\widetilde{C}_1$ we can now determine the last constant $C_1$ using \eqref{E:abm_2}
\begin{equation*}
    C_1=\widetilde{C}_1-m B_0^{1-p_1}\left(\frac{1+p_2}{\delta}-\frac{p_2}{r}\right).
\end{equation*}
As $\widetilde{C}_1<0$ and $p_2\delta<(1+p_2)r$, we deduce that $C_1<0$. Since we showed that the constants $C_1$, $\widetilde{C}_1$, $\widetilde{C}_2$ are each negative, we can conclude the candidate value function is concave $(0,\infty)$. Also note that, $V'(0)=m/\delta\le 1$. 
This shows that $V(h)\le \min(h,B_0)$ and $V$ satisfies all conditions of Proposition \ref{P:verification}.

Second, let us assume $m>\delta$. In this case, we postulate that there are two prepayment boundaries 
$h_1,h_2$ such that $h_1<B_0<h_2$. The value matching and smooth pasting at $h_2$ imply
\begin{equation*}
\begin{split}
   &\wt C_1 h_2^{p_1} + \wt C_2 h_2^{-p2}+mB_0/r=B_0\\
  & \wt C_1 p_1 h_2^{p_1-1} - \wt C_2 p_2 h_2^{-p2-1}=0
\end{split}
\end{equation*}
so that
\begin{equation}\label{abm:wtC12}
\begin{split}
&\wt{C}_{1}=-\frac{p_{2}}{p_1+p_2} m B_{0}\left(1/r-1 / m\right) h_{2}^{-p_{1}}<0\\
&\wt{C}_{2}=-\frac{p_{1}}{p_1+p_2} m B_{0}\left(1/r-1 / m\right) h_{2}^{p_{2}}<0
\end{split}
\end{equation}
as $m>r$.  Now the value matching and smooth pasting at $h_1$ give
\begin{equation*}
\begin{split}
   & C_1 h_1^{p_1} + C_2 h_1^{-p2}+m h_1/\delta=h_1\\
  &C_1 p_1 h_1^{p_1-1} - C_2 p_2 h_1^{-p2-1}+m /\delta=1
\end{split}
\end{equation*}
so that
\begin{equation}\label{abm:C12}
\begin{split}
&C_{1}=-\frac{1+p_{2}}{p_{1}+p_{2}}(m / \delta-1) h_{1}^{1-p_{1}}<0\\
&C_{2}=-\frac{p_{1}-1}{p_{1}+p_{2}}(m / \delta-1) h_{1}^{1+p_{2}}<0
\end{split}
\end{equation}
as $m>\delta$. Finally, imposing the  value matching and smooth pasting conditions at $B_0$ we get
\begin{equation*}
\begin{split}
&\widetilde{C}_1 B_{0}^{p_{1}}+\widetilde{C}_{2} B_{0}^{-p_{2}}+m B_{0} / r=C_{1} B_{0}^{p_{1}}+C_{2} B_{0}^{-p_{2}}+m B_{0} / \delta\\
&\widetilde{C}_1 p_1 B_{0}^{p_{1}-1}-\widetilde{C}_{2}p_2 B_{0}^{-p_{2}-1}=C_{1} p_1B_{0}^{p_{1}-1}-C_{2}p_2 B_{0}^{-p_{2}-1}+m/ \delta
\end{split}
\end{equation*}
that gives
\begin{equation*}
\begin{split}
&\widetilde{C}_{1}=C_{1}+\frac{1+p_{2}}{p_{1}\left(p_{1}+p_{2}\right)} \frac{m B_{0}^{1-p_1}}{\delta}\\
&\widetilde{C}_{2}=C_{2}-\frac{p_{1}-1}{p_{2}\left(p_{1}+p_{2}\right)} \frac{m B_{0}^{1+p_{2}}}{\delta}.
\end{split}
\end{equation*}
Now we insert \eqref{abm:wtC12} and \eqref{abm:C12} into two equations above and obtain the system of two algebraic equations for $h_1$ and $h_2$
\begin{equation*}
\begin{split}
&-\frac{1+p_{2}}{p_{1}+p_{2}}(m / \delta-1) h_{1}^{1-p_{1}}=-\frac{1+p_{2}}{p_{1}+p_{2}}(m / \delta-1) h_{1}^{1-p_{1}}+\frac{1+p_{2}}{p_{1}\left(p_{1}+p_{2}\right)} \frac{m B_{0}^{1-p_1}}{\delta}\\
&-\frac{p_{1}}{p_1+p_2} m B_{0}\left(1/r-1 / m\right) h_{2}^{p_{2}}=-\frac{p_{1}-1}{p_{1}+p_{2}}(m / \delta-1) h_{1}^{1+p_{2}}-\frac{p_{1}-1}{p_{2}\left(p_{1}+p_{2}\right)} \frac{m B_{0}^{1+p_{2}}}{\delta}.
\end{split}
\end{equation*}
Now after change of variables $h_1=xB_0$ and $h_2=yB_0$  this system can be rewritten as
\begin{equation}\label{abm:equation-xy}
\begin{split}
&\frac{-p_{2}}{1+p_{2}}\left(1/r-1 / m\right) y^{-p_{1}}=-\left(1/\delta-1 / m\right) x^{1-p_{1}}+\frac{1}{p_{1} \delta}\\
&-\frac{p_{1}}{p_{1}-1}(1 / r-1 / m) y^{p_{2}}=-(1 / \delta-1 / m) x^{1+p_{2}}-\frac{1}{p_{2} \delta}.
\end{split}
\end{equation}
The goal is to show that there exists unique pair $(x,y)$ that solves the system \eqref{abm:equation-xy} and $0<x<1<y$.
By rewriting the first equation we have
\begin{equation}\label{abm:equation-x}
    x=x(y)=\left(\frac{1 / \delta-1 / m}{\frac{1}{p_1 \delta}+\frac{p_{2}}{1+p_{2}}\left(1/r-1 / m\right) y^{-p_{1}}}\right)^{\frac{1}{p_{1}-1}}.
\end{equation}
We note that 
\begin{equation}\label{abm:condition-x}
x=x(y)<1 \qquad \Leftrightarrow \qquad y^{-p_{1}}>\frac{\left(\frac{p_{1}-1}{p_{1} \delta}-1 / m\right)}{\frac{p_{2}}{1+p_{2}}(1 / r-1 / m)}.
\end{equation}
Next, we insert \eqref{abm:equation-x} into the second equation of \eqref{abm:equation-xy}
\begin{equation*}
0 =\frac{p_{1}}{p_{1}-1}(1 / r-1 / m) y^{p_{2}}-\left(1/\delta-1/m\right) (x(y))^{1+p_{2}}-\frac{1}{p_{2} \delta} =:g(y)
\end{equation*}
and get the equation for $y$. The question is whether there is unique root to $g(y)=0$ such that $y>1$. Lengthy calculations show that $g(1)<0$. Now we differentiate  $g$ and obtain
\begin{equation}\label{abm:deriv-g}
   g'(y)= \frac{p_{1}p_{2}}{p_{1}-1}\left(1/r-1 / m\right) y^{-p_{1}-1}\left(y^{p_{1}+p_{2}}-(x(y))^{p_{1}+p_{2}}\right).
\end{equation}
Here, let us consider two cases: 1) $m\le p_1\delta/(p_1-1)$; 2) $m>p_1\delta/(p_1-1)$. 
In the former case, it is clear from \eqref{abm:condition-x} that $x=x(y)<1$ for all $y>1$. Also in this case, $g'(y)>0$ on $(1,\infty)$ as $x(y)<1<y$, and $g(\infty)=\infty$. Hence, we can conclude that there exists unique $\wh{y}>1$ such that $g(\wh{y})=0$ and $x(\wh{y})<1$, i.e., there exists unique solution $(\wh{x},\wh{y})$ to the system \eqref{abm:equation-xy} as desired. 

Now we turn the case $m>p_1\delta/(p_1-1)$. For this, we define
\begin{equation*}
    \bar y:=\left(\frac{\frac{p_{2}}{1+p_{2}} \left(\frac{1}{r}- \frac{1}{m}\right)}{\frac{p_{1}-1}{p_{1}} \frac{1}{\delta}-\frac{1}{m}}\right)^{1 / p_{1}}
\end{equation*}
so that $x(y)<1$ if and only if $y\in(1,\bar y)$. Then from \eqref{abm:deriv-g} we have that $g'(y)>0$ on $(1,\bar y)$. In addition, tedious algebra shows that $g(\bar y)>0$. This implies that there exists unique $\wh{y}\in(1,\bar y)$ such that $g(\wh{y})=0$ and $x(\wh{y})<1$. 

As value matching and smooth pasting are met, and the constants are negative, we conclude that $V$ is $C^1$, concave and that $V(h)\le \min[B_0,h]$ for $h>0$. Also, simple calculations show that $\Lcal_{H} V - r V + c(h)\ge 0$ for all $h>0$. Hence, we can invoke Proposition \ref{P:verification} to finish the proof.

\end{proof}


\subsection{APRM Proofs}

\begin{rem}\label{R:turnover_APRM} Similarly to the ABM, up to a linear term in $h$,  we may view the APRM with turnover, as an APRM without turnover, but with adjusted $r, \delta, m^P$ as well as $\alpha, B_0$. This is shown in  the following lemma.
\end{rem}

\begin{lem}\label{L:APRM_turnover_adjustment}
Write $V^P(h) = V^P(h; r,\delta,\lambda,\alpha,B_0,m^P)$ as the APRM value function, where $\lambda$ is the turnover intensity.  Then we have 
\begin{equation*}
    V^P(h; r,\delta,\lambda,\alpha,B_0,m^P) = \frac{\alpha\lambda}{\delta+\lambda}h + V^P\left(h; \wt{r},\wt{\delta},0,\wt{\alpha},\wt{B}_0,\wt{m}^P\right)
\end{equation*}
where $\wt{r} =  r+\lambda$, $\wt{\delta} = \delta + \lambda$, $\wt{\alpha} = \alpha\delta/(\delta+\lambda)$, $\wt{B}_0 =  B_0 - \alpha\lambda/(\lambda+\delta)$ and 
\begin{equation*}
    \wt{m}^P = (m^P+\lambda)\frac{B_0 - \frac{\alpha\lambda}{m^P+\lambda}}{B_0 - \frac{\alpha\lambda}{\delta+\lambda}}.
\end{equation*}
Additionally, the optimal stopping times for each value function coincide.
\end{lem}

\begin{proof}[Proof of Lemma \ref{L:APRM_turnover_adjustment}]
The lemma follows from three facts. First, $(x-1)^+ = x - \min\bra{1,x}, x>0$. Second, for any $t\geq 0$ 
\begin{equation*}
    e^{-(r+\lambda) t}H_t = h - (\lambda+\delta)\int_0^{t} H_u e^{-(r+\lambda)u}du + \sigma\int_0^t e^{-(r+\lambda)u} H_u dB_u.
\end{equation*}
Third, if we denote by $M$ the martingale on the right hand side above, then
\begin{equation*}
    -h \leq M_t \leq e^{-(r+\lambda) t}H_t + (\lambda+\delta)\int_0^{\infty}H_u e^{-(r+\lambda)u}du,
\end{equation*}
and $\cbra{e^{-(r+\lambda) t}H_t}_{t}$ is bounded in $L^{1+\eps}$ for $\eps>0$ small enough.  This allows us conclude that $\espalt{}{h}{M_{\tau}} = 0$ even for unbounded (possibly infinite) stopping times.   With these facts, the result follows by direct calculations.
\end{proof}

\begin{rem}\label{R:aprm_turnover_parameter_restrictions}

Recall our standing assumption that $m^P>r$.  For the adjustments in Lemma \ref{L:APRM_turnover_adjustment} to ensure $\wt{m}^{P} > \wt{r}$ we further require that
\begin{equation*}
    \alpha\lambda(\delta-r) < \frac{(m^P-r)B_0}{\delta+\lambda}.
\end{equation*}
This will always hold if $r\geq \delta$ or $m\geq \delta$. When $r < m < \delta$ this restriction is very mild;  for typical parameter values, it will hold provided $m$ exceeds $r$ by just a few basis points. Next, straightforward calculations show $\wt{\alpha} < \wt{B}_0 \Leftrightarrow \alpha < B_0$ and $\wt{m}\leq \wt{\delta} \Leftrightarrow m \leq \delta$. Lastly, recall from Remark \ref{R:alpha_star} that $\wt{\alpha} > p_2\wt{B}_0(\wt{m}^P/\wt{r}-1)$ ruled out a high state-prepayment region. Translating back to the original parameters we find
\begin{equation*}
    \begin{split}
    \wt{\alpha}>  p_2\wt{B}_0\left(\frac{\wt{m}^P}{\wt{r}}-1\right) \quad &\Longleftrightarrow \quad \frac{\alpha}{\lambda+\delta}\left((r+\lambda)\delta + p_2\lambda(\delta-r)\right) > p_2(m-r)B_0.
    \end{split}
\end{equation*}
Thus, when $(r+\lambda)\delta + p_2(\delta-r)\leq 0$, no matter how large $\alpha$ is one cannot rule out a high state prepayment region. This effect is not present absent turnover. 

\end{rem}

For the remainder of this section we assume $\lambda = 0$ and the parameters have been suitably adjusted.  To set notation,  we first amend Remark \ref{R:alpha_star}.

\begin{rem}\label{R:alpha_star_new}
Let us come back to the assumption that  $(1,\infty)$ is in the continuation region, so $V^P(h) = \wt{C}_2 h^{-p_2} + m^PB_0/r$ for $h>1$, where $\wt{C}_2<0$. The minimal difference in \eqref{E:min_diff_f} is $g(\alpha,-\wt{C}_2)$ where
\begin{equation}\label{E:g_def}
\begin{split}
g(\alpha;\beta) &\dfn (1+p_2)\left(\frac{\alpha}{p_2}\right)^{\frac{p_2}{1+p_2}}\beta^{\frac{1}{1+p_2}} - \alpha - B_0\left(\frac{m^P}{r}-1\right),\qquad \alpha,\beta > 0.
\end{split}
\end{equation}
$g(\alpha,\beta = -\wt{C}_2)$ must be non-negative. By considering $h\approx 1$ we require $\beta > B_0(m^P/r-1)$, and recall that we already know there is no high-state prepayment region when $\alpha \geq p_2B_0(m^P/r-1)$. Therefore, let us express $\alpha,\beta$ in terms of $B_0(m^P/r-1)$ by writing $\alpha = \upsilon \times p_2 B_0(m^P/r-1)$ and $\beta = \qtau \times B_0(m^P/r-1)$ for $0<\upsilon < 1, \qtau > 1$. Calculation then shows
\begin{equation*}
g(\alpha,\beta) = \wt{g}(\upsilon,\qtau) \dfn B_0\left(\frac{m^P}{r}-1\right)\left( (1+p_2)\qtau^{\frac{1}{1+p_2}} \upsilon^{\frac{p_2}{1+p_2}} - (1+p_2\upsilon)\right).
\end{equation*}
As $\qtau>1$ it is clear $\wt{g}(0,\qtau) < 0$, $\wt{g}(1,\qtau) > 0$ and $\upsilon \to \wt{g}(\upsilon,\qtau)$ is strictly increasing.  This gives, for each $\qtau>1$, a unique $\upsilon^* \in (0,1)$ (hence $\alpha^* \in (0,p_2B_0(m^p/r-1))$ solving $g(\upsilon^*,\qtau) = 0$.

Consider when $m<\delta$. The corresponding $\beta_1$ is
\begin{equation}\label{E:m_leq_delta_beta}
\beta_{1} \dfn -\wt{C}_2 = \frac{p_1-1}{p_2(p_1+p_2)}\frac{m^PB_0}{\delta} > B_0\left(\frac{m^P}{r}-1\right),
\end{equation}
where the last inequality follows using \eqref{E:p1_p2_ident}, as well as $m < \delta$.  According to the above, we may define
\begin{equation*}
\alpha^*\dfn \textrm{ the unique root of } g(\alpha,\beta_1) \textrm{ in } \left(0, p_2 B_0\left(\frac{m^P}{r}-1\right)\right) \subset \left(0,B_0\right),
\end{equation*}
where the last set inclusion follows from \eqref{E:mstar_first_use} and $m<\delta$.

Next, consider when $\delta < m^P < m^*$. The corresponding $\beta_2$ is now
\begin{equation}\label{E:m_geq_delta_beta}
\beta_{2} \dfn -\wt{C}_2 = \frac{p_1-1}{p_2(p_1+p_2)}\frac{m^P B_0}{\delta}\left(1 + p_2 p_1^{\frac{1+p_2}{p_1-1}}\left(1-\frac{\delta}{m^P}\right)^{\frac{p_1+p_2}{p_1-1}}\right).
\end{equation}
Here, using that $\delta < m^P < m^*$, a lengthy calculation shows $\beta_2 > B_0(m^P/r-1)$. As such we may define
\begin{equation*}
\alpha^*\dfn \textrm{ the unique root of } g(\alpha,\beta_2) \textrm{ in } \left(0, p_2 B_0\left(\frac{m^P}{r}-1\right)\right) \subset \left(0,B_0\right),
\end{equation*}
where the last set inclusion follows from \eqref{E:mstar_first_use} and $m<m^*$. 
\end{rem}

\begin{proof}[Proof of Proposition \ref{P:APRM_P1}]
Throughout, we drop all ``P'' superscripts. We also assume $m>r$, recall $p_1,p_2$ from \eqref{E:p1_p2} and the identity \eqref{E:p1_p2_ident}. First, consider when $\alpha < \alpha^*$.  Value matching and smooth pasting at $h_3$ yield
\begin{equation}\label{E:h2_def}
h_3 = \frac{p_2}{1+p_2}\left(\frac{mB_0}{\alpha}\left(\rcpdiff{r}{m}\right) + 1\right);\qquad \check{C}_2 = -\frac{\alpha}{p_2}h_3^{1+p_2} < 0.
\end{equation}
Note that $h_3 > 1$ since $\alpha < \alpha^* < p_2B_0 (m/r - 1)$. At $h_2$, value matching/smooth pasting give
\begin{equation}\label{E:eq1}
\begin{split}
\wt{C}_1 &= \frac{1+p_2}{p_1+p_2}\alpha h_2^{1-p_1} - \frac{p_2\alpha}{p_1+p_2}\left(\frac{mB_0}{\alpha}\left(\rcpdiff{r}{m}\right) + 1\right)h_2^{-p_1},\\
\wt{C}_2 &= \frac{p_1-1}{p_1+p_2}\alpha h_2^{1+p_2} - \frac{p_1\alpha}{p_1+p_2}\left(\frac{mB_0}{\alpha}\left(\rcpdiff{r}{m}\right) + 1\right)h_2^{p_2}.
\end{split}
\end{equation}
Similarly, value matching/smooth pasting at $1$, along with \eqref{E:p1_p2_ident}, give
\begin{equation}\label{E:eq2}
\wt{C}_1 = C_1 + \frac{(1+p_2)}{p_1(p_1+p_2)}\frac{mB_0}{\delta};\qquad \wt{C}_2 = -\frac{(p_1-1)}{p_2(p_1+p_2)}\frac{mB_0}{\delta}.
\end{equation}
Matching the equations for $\wt{C}_2$ in \eqref{E:eq1}, \eqref{E:eq2} we seek $h_2$ such that
\begin{equation*}
-\frac{(p_1-1)}{p_2}\frac{mB_0}{\delta} = (p_1-1)\alpha h_2^{1+p_2} - p_1\alpha\left(\frac{mB_0}{\alpha}\left(\rcpdiff{r}{m}\right) + 1\right)h_2^{p_2}
\end{equation*}
and $1 < h_2 < h_3$. Using the formula for $h_3$ and simplifying we are left to find a zero of
\begin{equation*}
\chi(h) \dfn  h^{p_2}\left(\frac{p_1}{p_1-1}\frac{1+p_2}{p_2} h_3 - h\right) - \frac{mB_0}{p_2\alpha\delta},
\end{equation*}
in $1 < h_2 < h_3$. It is easy to see $\chi$ is increasing for $1 < h < h_3$ and calculation shows (recall $g$ from \eqref{E:g_def} and $\beta_1$ from \eqref{E:m_leq_delta_beta}) 
\begin{equation*}
\begin{split}
\chi(1) < 0 &\Longleftrightarrow \alpha < mB_0\left(\rcpdiffb{p_1}{m}{p_1-1}{\delta}\right);\qquad \chi(h_3) > 0 \Longleftrightarrow g(\alpha;\beta_1) < 0.
\end{split}
\end{equation*}
Now, $g(0,\beta_1) < 0$, $\alpha < \alpha^*$ and $g(\alpha^*,\beta_1) = 0$ so we know $g(\alpha;\beta_1) < 0$ and hence $\chi(h_3) > 0$.  As for $\chi(1)$, $m\leq \delta$ implies (using \eqref{E:p1_p2_ident})
\begin{equation*}
\rcpdiffb{p_1}{m}{p_1-1}{\delta} - p_2\left(\rcpdiff{r}{m}\right) = \frac{p_1+p_2}{m} - \frac{p_1-1}{\delta}\left(1 + \frac{1+p_2}{p_1}\right) \geq \frac{1+p_2}{p_1\delta} > 0.
\end{equation*}
Therefore, if $\alpha < p_2mB_0(1/r-1/m)$ then $\alpha < mB_0((p_1/m)-(p_1-1)/\delta)$, and there exists a unique $h_2 \in (1,h_3)$ satisfying $\chi(h_2) = 0$. At this point, we have identified $1 < h_2 < h_3$ and shown $\check{C}_2 < 0$, $\wt{C}_2 < 0$.  The last thing to do is show $C_1 < 0$, $\wt{C}_1 < 0$.  First, from \eqref{E:h2_def} and \eqref{E:eq1} we have
\begin{equation*}
\wt{C}_1 = -\frac{1+p_2}{p_1+p_2}\alpha h_1^{-p_1}\left(h_3 - h_2\right) < 0.
\end{equation*}
Next, using \eqref{E:eq2} we also see $C_1 < 0$.

As value matching and smooth pasting are met, and the constants are negative, we conclude that $V$ is $C^1$ and concave.  As $0 \leq V'(h)\rightarrow 0$ for $h\rightarrow \infty$ we know that $V$ is non-decreasing, and
\begin{equation}\label{E:deriv_bound}
0 \leq V'(h) \leq \textrm{Constant}\times \left(1_{h\leq h_3} + h^{-p_2-1}1_{h>h_3}\right) \leq \textrm{Constant}\times \left(1_{h\leq h_3} + h^{-1}1_{h>h_3}\right)
\end{equation}
where the last equality holds as $h_3 > 1$. From here, it is easy to see that for all $h>0$ that \eqref{E:dotV_integ} holds. Therefore, we may invoke Proposition \ref{P:verification} provided
\begin{enumerate}[(a)]
\item $V(h) \leq B_0\min\bra{1,h} + \alpha(h-1)^+$;
\item $mB_0 - \alpha\delta h_2 - r(B_0-\alpha) \geq 0$.
\end{enumerate}
To show (a), note that $V'(0) = mB_0/\delta < B_0$. Thus by concavity, $V(h)\leq B_0h$ on $[0,1]$.  As $V(h) = B_0 + \alpha(h-1)$ on $[h_2,h_3]$ by concavity we know $V(h)\leq B_0 + \alpha(h-1)$ on $[h_3,\infty)$.  As for the interval $(1,h_2)$, if there were a point $h_0 \in (1,h_2)$ with $V(h_0) = B_0 + \alpha(h_0 - 1)$ then since $V_h > \alpha$ on $(h_0,h_2)$ we cannot have $V(h_2) = B_0 + \alpha(h_2 - 1)$, proving (a).  As for (b), we note by \eqref{E:p1_p2_ident} and \eqref{E:h2_def}
\begin{equation*}
h_3 = \frac{p_2}{\alpha r(1+p_2)}\left(B_0(m-r) + \alpha r\right) = \frac{p_1-1}{\alpha\delta p_1}\left(B_0(m-r) + \alpha r\right).
\end{equation*}
From here it is clear $B_0(m-r) + \alpha r > \alpha\delta h_3$. This gives the result for $\alpha < \alpha^*$.

We next consider case $\alpha\geq \alpha^*$.  Value matching and smooth pasting at $1$, along with \eqref{E:p1_p2_ident} and \eqref{E:m_leq_delta_beta} give
\begin{equation}\label{E:K1wtK2}
K_1 = -\frac{1+p_2}{p_1(p_1+p_2)}\frac{mB_0}{\delta} < 0\qquad \wt{K}_2 = -\frac{p_1-1}{p_2(p_1+p_2)}\frac{mB_0}{\delta} = -\beta_1 < 0.
\end{equation}
This ensures $V$ is $C^1$, concave, and since $0\leq \dot{V}(h) \rightarrow 0$ as $h\rightarrow 0$ $V$ is also non-decreasing, with derivative which also satisfies \eqref{E:deriv_bound} (with potentially different constant) and hence \eqref{E:dotV_integ}.  Therefore, since $(\mathcal{L}-r)V + c = 0$ on the entire region, we may invoke Proposition \ref{P:verification} provided $V(h) \leq B_0\min\bra{1,h} + \alpha(h-1)^+$. Concavity implies $V(h) \leq B_0 h$ on $(0,1]$ since $\dot{V}(0) = mB_0/\delta \leq B_0$. As for $h>1$ consider the function (c.f. Remarks \ref{R:alpha_star}, \ref{R:alpha_star_new})
\begin{equation*}
\chi(h) \dfn B_0 + \alpha(h-1) - V(h) = B_0 + \alpha(h-1) - \wt{K}_2 h^{-p_2} - \frac{mB_0}{r}.
\end{equation*}
We know $\chi(1) > 0$, and $\chi$ is strictly convex with derivative $\chi_h(h) = \alpha + p_2\wt{K}_2 h^{-1-p_2}$. Therefore, if $\alpha \geq -p_2\wt{K}_2 > 0$ then $\dot{\chi}(1)\geq 0$ and by convexity, $\chi>0$ on $(1,\infty)$.  Else the unique minimum of $\chi$ occurs at $h_0 = \left(-\alpha/(p_2\wt{K}_2)\right)^{-1/(1+p_2)} > 1$. Plugging $h_0$ into $\chi$, using \eqref{E:K1wtK2}, and simplifying shows that $\chi(h_0) = g(\alpha;\beta_1)$.  Now, by assumption we have $\alpha^* \leq \alpha < -p_2\wt{K}_2$ and that $\alpha^*$ is the unique $0$ of $g(\cdot;\beta_1)$ on $(0,p_2mB_0(1/r-1/m))$. Unfortunately, calculation shows $p_2mB_0(1/r-1/m) < -p_2\wt{K}_2$ so we do not immediately know (recall $g(0,\beta_1) < 0$) that $g(\alpha,\beta_1) > 0$. However, from \eqref{E:g_def}, \eqref{E:K1wtK2} we see
\begin{equation*}
g(\alpha,\beta_1) = \frac{1+p_2}{p_2}\left(-p_2\wt{K}_2\right)^{\frac{1}{1+p_2}} \alpha^{\frac{p_2}{1+p_2}} - \alpha - mB_0\left(\rcpdiff{r}{m}\right)
\end{equation*}
and hence $g'(\alpha) > 0$ on $(0,-p_2\wt{K_2})$. Since $g(0)<0$ and
\begin{equation*}
g(-p_2\wt{K}_2) = -\wt{K}_2 - mB_0\left(\rcpdiff{r}{m}\right) > 0
\end{equation*}
there is a unique $\wh{\alpha}$ on $(0,-p_2\wt{K}_2)$ with $g(\wh{\alpha}) = 0$. Thus, by our assumption of a unique $0$ over the smaller $(0,p_2mB_0(1/r-1/m))$ it must be that $\wh{\alpha} = \alpha^*$ and hence as $\alpha \geq \alpha^*$ it follows that $g(\alpha) \geq 0$ and hence $V(h) \leq B_0 + \alpha(h-1)$ on $(1,\infty)$, giving the verification result for $\alpha \geq \alpha^*$ and finishing the proof.

\end{proof}


\begin{proof}[Proof of Proposition \ref{P:APRM_P2}] We drop the ``P'' superscripts, assume $m>r$, recall $p_1,p_2$ from \eqref{E:p1_p2} and heavily use \eqref{E:p1_p2_ident}. This proof is significantly more involved than that of Proposition  \ref{P:APRM_P1}, and we will start with the $\alpha > \alpha^*$ case. Value matching and smooth pasting at $1$, along with \eqref{E:p1_p2_ident} give
\begin{equation}\label{E:mid_m_1}
K_1 = -\frac{1+p_2}{p_1(p_1+p_2)}\frac{mB_0}{\delta}< 0;\qquad \wt{K_2}-K_2 = -\frac{p_1-1}{p_2(p_1+p_2)}\frac{mB_0}{\delta}.
\end{equation}
Value matching and smooth pasting at $h_1$ give
\begin{equation}\label{E:mid_m_2}
K_1 = -\frac{1+p_2}{p_1+p_2}\frac{mB_0}{\delta} h_1^{1-p_1}\left(1-\delmf\right);\qquad K_2 = -\frac{p_1-1}{p_1+p_2}\frac{mB_0}{\delta} h_1^{1+p_2}\left(1-\delmf\right).
\end{equation}
Using the equations for $K_1$ in \eqref{E:mid_m_1} and \eqref{E:mid_m_2} we deduce
\begin{equation*}
h_1 = \left(p_1\left(1-\delmf\right)\right)^{\frac{1}{p_1-1}}.
\end{equation*}
Note that $\delta < m < p_1\delta/(p_1-1)$ implies $0 < h_1 < 1$.  Using this expression for $h_1$ along with \eqref{E:mid_m_1}, \eqref{E:mid_m_2} we conclude (recall \eqref{E:m_geq_delta_beta})
\begin{equation*}
\begin{split}
K_2 &= -\frac{p_1-1}{p_1+p_2}\frac{mB_0}{\delta} p_1^{\frac{1+p_2}{p_1-1}}\left(1-\delmf\right)^{\frac{p_1+p_2}{p_1-1}} < 0\\
\wt{K}_2 &= -\frac{p_1-1}{p_1+p_2}\frac{mB_0}{\delta}\left(\frac{1}{p_2} + p_1^{\frac{1+p_2}{p_1-1}}\left(1-\delmf\right)^{\frac{p_1+p_2}{p_1-1}}\right) = -\beta_2 < 0.
\end{split}
\end{equation*}
This shows that $V$ is $C^1$, concave, increasing with $V\leq B_0h$ on $[0,1]$ and with \eqref{E:dotV_integ} holding.  Thus, Proposition \ref{P:verification} will yield the result, provided 
\begin{enumerate}[(a)]
\item $V\leq B_0 + \alpha(h-1)$ on $(1,\infty)$;
\item $(\mathcal{L}-r)V + mB_0h$ on $[0,h_1]$.  
\end{enumerate}

To show (a), (c.f. Remarks \ref{R:alpha_star}, \ref{R:alpha_star_new}) define the strictly convex function on $(1,\infty)$
\begin{equation*}
\chi(h) \dfn B_0 + \alpha(h-1) - \wt{K}_2 h^{-p_2} - \frac{mB_0}{r} = B_0 + \alpha(h-1) + \beta_2 h^{-p_2} - \frac{mB_0}{r}.
\end{equation*}
As $\chi'(h) = \alpha - p_2\beta_2$, if $\alpha \geq p_2\beta_2$ then $\chi(1)\geq 0$ implies $\chi\geq 0$ on $(1,\infty)$.  Else (recall $\alpha^* \leq p_2mB_0(1/r-1/m) < p_2 \beta_2$) we have $\alpha^* < \alpha < p_2\beta_2$ and $\chi$ is minimized at $h_0 = (\alpha/(p_2\beta_2))^{-1/(1+p_2)}$.  Plugging this value in $\chi$, calculation shows $\chi(h_0) = g(\alpha,\beta_2)$.  But, as shown in the proof of Proposition  \ref{P:APRM_P1}, $g(\cdot;\beta_2)$ is increasing on $(0,p_2\beta_2)$.  Thus, $\alpha>\alpha^*$ implies that on $(1,\infty)$, $\chi(h) \geq \chi(h_0) = g(\alpha,\beta_2) > 0$. To show b), note that on $[0,b_1]$, $(\mathcal{L}-r)V+mB_0 = (m-\delta)B_0 h \geq 0$. This finishes the proof when $\alpha > \alpha^*$.

We next turn to $\alpha \leq \alpha^*$.  Value matching and smooth pasting at $h_3$ imply
\begin{equation*}
h_3 = \frac{p_2}{1+p_2}\left(\frac{mB_0}{\alpha}\left(\rcpdiff{r}{m}\right)+1\right),
\end{equation*}
and $h_3 > 1$ since $\alpha \leq \alpha^* < p_2mB_0(1/r-1/m)$. Furthermore $\check{C}_2 = -(\alpha/p_2) h_3^{1+p_2} < 0$. Next, value matching and smooth pasting at $h_2$ gives
\begin{equation}\label{E:mid_m_h2_AB}
\begin{split}
\wt{C}_1 &= \frac{1+p_2}{p_1+p_2}\alpha h_2^{1-p_1}-\frac{p_2}{p_1+p_2}\alpha h_2^{-p_1}\left(\frac{mB_0}{\alpha}\left(\rcpdiff{r}{m}\right) + 1\right)\\
\wt{C}_2 &= \frac{p_1-1}{p_1+p_2}\alpha h_2^{1+p_2}-\frac{p_1}{p_1+p_2}\alpha h_2^{p_2}\left(\frac{mB_0}{\alpha}\left(\rcpdiff{r}{m}\right) + 1\right)
\end{split}
\end{equation}
Continuing, value matching and smooth pasting at $1$ along with \eqref{E:p1_p2_ident} give
\begin{equation}\label{E:mid_m_h2_AB2}
\begin{split}
\wt{C}_1 &= C_1 + \frac{1+p_2}{p_1(p_1+p_2)}\frac{mB_0}{\delta};\qquad \wt{C}_2 = C_2 - \frac{p_1-1}{p_2(p_1+p_2)}\frac{mB_0}{\delta}.
\end{split}
\end{equation}
Lastly, value matching and smooth pasting at $h_1$ give
\begin{equation}\label{E:mid_m_h1_AB}
\begin{split}
C_1 &= -\frac{1+p_2}{p_1+p_2}\frac{mB_0}{\delta} h_1^{1-p_1}\left(1-\delmf\right) < 0;\qquad C_2 = -\frac{p_1-1}{p_1+p_2}\frac{mB_0}{\delta} h_1^{1+p_2}\left(1-\delmf\right) < 0.
\end{split}
\end{equation}
Using the equations for $\wt{C}_1$ in \eqref{E:mid_m_h2_AB}, \eqref{E:mid_m_h2_AB2}, and plugging in for $C_1$ from \eqref{E:mid_m_h1_AB} we obtain after some simplifications
\begin{equation}\label{E:mid_m_h1_eqn1}
\frac{mB_0}{\alpha\delta}\left(h_1^{1-p_1}p_1\left(1-\delmf\right) - 1\right) = p_1h_2^{-p_1}\left(h_3-h_2\right).
\end{equation}
Solving for $h_1$ yields
\begin{equation}\label{E:mid_m_h1_eqn2}
h_1 = \left(\frac{p_1\left(1-\delmf\right)}{1+\frac{\alpha\delta}{mB_0}p_1 h_2^{-p_1}(h_3-h_2)}\right)^{\frac{1}{p_1-1}}.
\end{equation}
Note that $\delta < m < p_1\delta/(1-p_1)$, $p_1 > 1$ and $h_2 < h_3$ (as we will show) implies $0 < h_1 < 1$. Next, using the equations for $\wt{C}_2$ in \eqref{E:mid_m_h2_AB}, \eqref{E:mid_m_h2_AB2},  and plugging in for $C_1$ from \eqref{E:mid_m_h1_AB} we obtain after some simplifications
\begin{equation}\label{E:mid_m_h2_eqn1}
\frac{mB_0}{\alpha\delta}\left(h_1^{1+p_2}p_2\left(1-\delmf\right)+1\right) = p_2h_2^{p_2}\left(\frac{p_1(1+p_2)}{(p_1-1)p_2}h_3-h_2\right).
\end{equation}
Plugging in for $h_1$ from \eqref{E:mid_m_h1_eqn2} yields the main equation
\begin{equation}\label{E:mid_m_h2_eqn2}
\frac{mB_0}{\alpha\delta}\left(p_2\left(1-\delmf\right)\left(\frac{p_1\left(1-\delmf\right)}{1+\frac{\alpha\delta}{mB_0}p_1 h_2^{-p_1}(h_3-h_2)}\right)^{\frac{1+p_2}{p_1-1}} +1\right) = p_2h_2^{p_2}\left(\frac{p_1(1+p_2)}{(p_1-1)p_2}h_3-h_2\right),
\end{equation}
and our goal is to show there is a unique solution $h_2$ lying in $(1,h_3)$ for all $0 < \alpha < \alpha^*$. After dividing by $p_2$, define the function
\begin{equation}\label{E:mid_m_chi_def}
\begin{split}
\chi(h) \dfn& h^{p_2}\left(h -\frac{p_1(1+p_2)}{(p_1-1)p_2}h_3\right) + \frac{mB_0}{\alpha\delta}\left(\left(1-\delmf\right)\left(\frac{p_1\left(1-\delmf\right)}{1+\frac{\alpha\delta}{mB_0}p_1 h^{-p_1}(h_3-h)}\right)^{\frac{1+p_2}{p_1-1}} +\frac{1}{p_2}\right)\\
=& h^{1+p_2} - \frac{p_1}{p_1-1}\left(\frac{mB_0}{\alpha}\left(\rcpdiff{r}{m}\right)+1\right)h^{p_2}  \\
&+ \frac{mB_0}{\alpha\delta}\left(\ \frac{1}{p_2} + \left(1-\delmf\right)\left(\frac{ p_1 \left(1-\delmf\right)}{1+ \frac{\alpha\delta}{mB_0} p_1 h^{-p_1}\left(\frac{p_2}{1+p_2}\left(\frac{mB_0}{\alpha}\left(\rcpdiff{r}{m}\right)+1\right)- h\right)}\right)^{\frac{1+p_2}{p_1-1}} \right)
\end{split}
\end{equation}
Here, we have defined $\chi$ using both $h_3$ and its formula, as both forms will be needed below.  Though the algebra below is tricky, the idea of the proof is simple. We will show that
\begin{enumerate}[(1)]
\item $\chi(1) > 0$;
\item $\chi$ is strictly decreasing;
\item For $\alpha \leq \alpha^*$, $\chi(h_3) < 0$.
\end{enumerate}
This will provide the desired unique solution.  To this end
\begin{equation*}
\begin{split}
\chi(1) &= 1 - \frac{p_1}{p_1-1}\left(\frac{mB_0}{\alpha}\left(\rcpdiff{r}{m}\right)+1\right)  \\
&+ \frac{mB_0}{\alpha\delta}\left(\ \frac{1}{p_2}+ \left(1-\delmf\right)\left(\frac{ \left(1-\delmf\right)}{\frac{1}{p_1}+ \frac{\alpha\delta}{mB_0} \left(\frac{p_2}{1+p_2}\left(\frac{mB_0}{\alpha}\left(\rcpdiff{r}{m}\right)+1\right)- 1\right)}\right)^{\frac{1+p_2}{p_1-1}} \right).
\end{split}
\end{equation*}
Using \eqref{E:p1_p2_ident} one can show
\begin{equation*}
\begin{split}
1-\frac{p_1}{p_1-1}\left(\frac{mB_0}{\alpha}\left(\rcpdiff{r}{m}\right)+1\right)+\frac{mB_0}{ p_2\alpha\delta} &= \frac{mB_0}{\alpha\delta}\left(\frac{1}{p_1-1}\frac{\delta}{m}\left(1-\frac{\alpha}{B_0}\right) - \left(1-\delmf\right)\right)\\
\frac{1}{p_1} + \frac{\alpha\delta}{mB_0} \left(\frac{p_2}{1+p_2}\left(\frac{mB_0}{\alpha}\left(\rcpdiff{r}{m}\right)+1\right)- 1\right) &= 1-\delmf + \frac{1}{1+p_2}\frac{\delta}{m}\left(1-\frac{\alpha}{B_0}\right).
\end{split}
\end{equation*}
This gives
\begin{equation*}
\chi(1) = \frac{mB_0}{\alpha\delta}\wh{\chi}\left(1-\delmf,\delmf\left(1-\frac{\alpha}{B_0}\right)\right)
\end{equation*}
where
\begin{equation*}
\wh{\chi}(x,y) \dfn \frac{y}{p_1-1} - x\left(1-\left(\frac{x}{x+\frac{1}{1+p_2}y}\right)^{\frac{1+p_2}{p_1-1}}\right).
\end{equation*}
Note that $\wh{\chi}(x,0) = 0$ and
\begin{equation*}
\partial_y \wh{\chi}(x,y) = \frac{1}{p_1-1}\left(1 - \left(\frac{x}{x+\frac{1}{1+p_2}y}\right)^{\frac{1+p_2}{p_1-1}+1}\right) > 0 \quad \textrm{on}\quad y > 0.
\end{equation*}
This gives $\chi(1) > 0$ since $\alpha < \alpha^* < p_2 mB_0(1/r-1/m)$ implies (again using \eqref{E:p1_p2_ident})
\begin{equation}\label{E:mid_m_alpha_leq_B0}
\frac{\alpha}{B_0} < p_2\left(\frac{m}{r}-1\right) < p_2\left(\frac{p_1\delta}{(p_1-1)r}  -1\right) = 1.
\end{equation}
We next show $\chi$ is decreasing on $(1,h_3)$.  Using the representation of $\chi$ with $h_3$ we have
\begin{equation*}
\begin{split}
\chi'(h) &= (1+p_2)h^{p_2} - \frac{p_1(1+p_2)}{p_1-1}h_3 h^{p_2-1}\\
&\qquad  + \frac{mB_0}{\alpha\delta}\left(1-\delmf\right)\frac{1+p_2}{p_1-1}\left(\frac{ p_1 \left(1-\delmf\right)}{1+ \frac{\alpha\delta}{mB_0} p_1 h^{-p_1}\left(h_3- h\right)}\right)^{\frac{1+p_2}{p_1-1}-1}\\
&\qquad \times \frac{-p_1\left(1-\delmf\right)}{\left(1+ \frac{\alpha\delta}{mB_0} p_1 h^{-p_1}\left(h_3- h\right)\right)^2}\times \frac{\alpha\delta}{mB_0} p_1\left(-p_1 h^{-p_1-1}h_3 + (p_1-1)h^{-p_1}\right)\\
&= (1+p_2)h^{p_2-1}\left(h-\frac{p_1}{p_1-1}h_3\right)\\
&\qquad + (1+p_2)\left(\frac{ p_1 \left(1-\delmf\right)}{1+ \frac{\alpha\delta}{mB_0} p_1 h^{-p_1}\left(h_3- h\right)}\right)^{\frac{p_1+p_2}{p_1-1}}h^{-p_1-1}\left(\frac{p_1}{p_1-1}h_3-h\right)\\
&= -(1+p_2) h^{p_2-1}\left(\frac{p_1}{p_1-1}h_3-h\right)\\
&\qquad \times \left(1-\left(h^{p_1-1}\right)^{-\frac{p_1+p_2}{p_1-1}}\left(\frac{ p_1 \left(1-\delmf\right)}{1+ \frac{\alpha\delta}{mB_0} p_1 h^{-p_1}\left(h_3- h\right)}\right)^{\frac{p_1+p_2}{p_1-1}}\right)\\
&= -(1+p_2) h^{p_2-1}\left(\frac{p_1}{p_1-1}h_3-h\right)\left(1-\left(\frac{ p_1 \left(1-\delmf\right)h}{h^{p_1}+ \frac{\alpha\delta}{mB_0} p_1\left(h_3- h\right)}\right)^{\frac{p_1+p_2}{p_1-1}}\right)\\
\end{split}
\end{equation*}
Therefore, $\chi' < 0$ if and only if
\begin{equation}\label{E:mid_m_chi_decr_cond}
h^{p_1} + \frac{\alpha\delta}{mB_0} p_1\left(h_3- h\right) - p_1\left(1-\delmf\right) h > 0.
\end{equation}
But, on $1 < h < h_3$ we have, since $m < p_1\delta/(p_1-1)$
\begin{equation*}
\begin{split}
h^{p_1} + \frac{\alpha\delta}{mB_0} p_1\left(h_3- h\right) - p_1\left(1-\delmf\right) h > h^{p_1} - p_1\left(1-\delmf\right) h = h^{p_1}- h > 0.
\end{split}
\end{equation*}
Therefore, $\chi$ is strictly decreasing on $(1,h_3)$. It remains to show $\chi(h_3) < 0$ for $\alpha < \alpha^*$.  We have
\begin{equation*}
\begin{split}
\chi(h_3) &= h_3^{1+p_2} - \frac{p_1}{p_1-1}\left(\frac{mB_0}{\alpha}\left(\rcpdiff{r}{m}\right) + 1\right)h_3^{p_2}\\
&\qquad + \frac{mB_0}{\alpha\delta}\left(\frac{1}{p_2} + \left(1-\delmf\right)\left(p_1\left(1-\delmf\right)\right)^{\frac{1+p_2}{p_1-1}}\right)\\
&= -\left(\frac{p_2}{1+p_2}\right)^{p_2}\left(\frac{mB_0}{\alpha}\left(\rcpdiff{r}{m}\right) + 1\right)^{1+p_2}\frac{p_1+p_2}{(1+p_2)(p_1-1)}\\
&\qquad + \frac{mB_0}{\alpha\delta}\left(\frac{1}{p_2} + \left(1-\delmf\right)\left(p_1\left(1-\delmf\right)\right)^{\frac{1+p_2}{p_1-1}}\right)
\end{split}
\end{equation*}
Note that $(Ax+1)^{1+p_2} = x\left(Ax^{p_2/(1+p_2)} + x^{-1/(1+p2)}\right)$. This gives that $\chi(h_3) < 0$ if and only if
\begin{equation}\label{E:alt_chi_eqn}
\begin{split}
&\left(\frac{p_2}{1+p_2}\right)^{p_2}\frac{p_1+p_2}{(1+p_2)(p_1-1)}\left(\left(\frac{mB_0}{\alpha} \right)^{\frac{p_2}{1+p_2}}\left(\rcpdiff{r}{m}\right) + \left(\frac{mB_0}{\alpha}\right)^{-\frac{1}{1+p_2}}\right)^{1+p_2}\\
&\qquad  - \frac{1}{\delta}\left(\frac{1}{p_2} + \left(1-\delmf\right)\left(p_1\left(1-\delmf\right)\right)^{\frac{1+p_2}{p_1-1}}\right) > 0.
\end{split}
\end{equation}
\begin{rem}\label{R:alt_chi_rem}
The map $x\to Ax^{p/(1+p)} + x^{-1/(1+p)}$ is increasing on $x > 1/(pA)$.  Applied to $x = (mB_0)/\alpha$, $A=1/r-1/m$ and $p=p_2$ we see that indeed $x>1/(pA)$ because $\alpha < p_2 mB_0(1/r-1/m)$.  This implies the left-hand side above is strictly decreasing in $\alpha$, and the inequality certainly holds for $\alpha\rightarrow 0$.
\end{rem}
In light of the above remark, let us plug in $\alpha = p_2 mB_0(1/r-1/m)$ to the above.  A number of simplifications occur and the left side becomes
\begin{equation*}
\begin{split}
&\frac{p_1+p_2}{p_1-1}\left(\rcpdiff{r}{m}\right)  - \frac{1}{\delta}\left(\frac{1}{p_2} + \left(1-\delmf\right)\left(p_1\left(1-\delmf\right)\right)^{\frac{1+p_2}{p_1-1}}\right).
\end{split}
\end{equation*}
Using \eqref{E:p1_p2_ident} one can show
\begin{equation*}
\frac{p_1+p_2}{p_1-1}\left(\frac{\delta}{r}-\delmf\right) - \frac{1}{p_2} = \frac{1}{p_1} + \frac{p_1+p_2}{p_1}\left(1-\frac{p_1}{p_1-1}\frac{\delta}{m}\right).
\end{equation*}
Therefore, the positivity of the above is equivalent to
\begin{equation*}
\psi\left(\frac{\delta}{m}\right) > 0,\qquad \psi(x) \dfn \frac{1}{p_1} + \frac{p_1+p_2}{p_1}\left(1-\frac{p_1}{p_1-1}x\right) - p_1^{\frac{1+p_2}{p_1-1}}\left(1-x\right)^{\frac{p_1+p_2}{p_1-1}}.
\end{equation*}
It is easy to see that $\psi'(x) < 0 \Leftrightarrow x > p_1/(p_1-1)$, but this is precisely the condition we have upon $\delta/m$.  Therefore, $\psi(\delta/m) > \psi((p_1-1)/p_1) = 0$ and hence $\chi(h_3) > 0$ always when $\alpha = p_2mB_0(1/r-1/m)$, which is not what we want.  However, in light of Remark \ref{R:alt_chi_rem}, we know $\chi(h_3) < 0$ as $\alpha\rightarrow 0$ and that there is a unique $0 < \wh{\alpha} < p_2 mB_0(1/r-1/m)$ such that $\chi(h_3) < 0$ if and only if $0 < \alpha < \wh{\alpha}$.  We now proceed to show that $\wh{\alpha} = \alpha^*$.

Recall from \eqref{E:mid_m_h1_eqn2} that for $h_2 = h_3(\alpha)$ we have $h_1 = (p_1(1-\delta/m))^{1/(p_1-1)}$. Plugging these values in \eqref{E:mid_m_h2_eqn2} at $\alpha = \wh{\alpha}$ (which enforces \eqref{E:mid_m_h2_eqn2}) we have (after dividing by $p_2$)
\begin{equation*}
\left(\frac{p_1+p_2}{p_2(p_1-1)}\right)h_3^{1+p_2} =
\frac{mB_0}{\alpha\delta}\left(\frac{1}{p_2} + p_1^{\frac{1+p_2}{p_1-1}}\left(1-\delmf\right)^{\frac{p_1+p_2}{p_1-1}}\right) = \frac{p_1+p_2}{\alpha(p_1-1)}\beta_2,
\end{equation*}
where the last equality holds from \eqref{E:m_geq_delta_beta}.  This is equivalent to
\begin{equation*}
(p_2\beta_2)^{\frac{1}{1+p_2}} = \alpha^{\frac{1}{1+p_2}} h_3 = \alpha^{-\frac{p_2}{1+p_2}}\frac{p_2}{1+p_2}\left(mB_0\left(\rcpdiff{r}{m}\right) + \alpha\right),
\end{equation*}
or
\begin{equation*}
0 = \frac{1+p_2}{p_2}(p_2\beta_2)^{\frac{1}{1+p_2}}\alpha^{\frac{p_2}{1+p_2}} - \alpha - (mB_0\left(\rcpdiff{r}{m}\right) = g(\alpha,\beta_2).
\end{equation*}
This shows $\wh{\alpha} = \alpha^*$ and finishes the proof showing a unique solution $h_2$ to \eqref{E:mid_m_h2_eqn2} in the interval $(1,h_3)$ for $0 < \alpha < \alpha^*$.

We now complete the remainder of the proof for $\alpha < \alpha^*$.  We have established that $V$ is $C^1$.  We have also shown $\check{C}_2 < 0$ and from \eqref{E:mid_m_h1_AB} we know $C_1,C_2 < 0$. In light of \eqref{E:mid_m_h2_AB2} we know $\wt{C}_2 < 0$.  As for $\wt{C}_1$, from \eqref{E:mid_m_h2_AB2} and \eqref{E:mid_m_h1_AB} we see
\begin{equation}\label{E:mid_m_wtC1}
\begin{split}
\wt{C}_1 &= -\frac{1+p_2}{p_1+p_2}\frac{mB_0}{p_1\delta}\left(p_1h_1^{1-p_1}\left(1-\delmf\right) - 1\right)\\
&= -\frac{1+p_2}{p_1+p_2}\frac{mB_0}{p_1\delta}\left(\frac{1+\frac{\alpha\delta}{mB_0}h_2^{-p_1}(h_3- h_2)}{1-\delmf} - 1\right) < 0,
\end{split}
\end{equation}
since $0 < 1-\delta/m < 1/p_1 < 1$.    This ensures $V$ is concave, and clearly \eqref{E:dotV_integ} holds.  Concavity implies $V(h) \leq B_0 h$ on $[0,1]$, and $V(h) \leq B_0 + \alpha(h-1)$ on $[h_2,\infty)$.  It also implies $V(h) \leq B_0 + \alpha(h-1)$ on $[1,h_2]$ since $V(h_2) = B_0 + \alpha(h_2-1)$ and $\dot{V} \geq \alpha$ on $[1,h_2]$.  As $(\mathcal{L}-r)V + mB_0\min\bra{1,h} = 0$ on $(h_1,h_2)$ and $(h_3,\infty)$ we need only show this is true on $[0,h_1]$ and $[h_2,h_3]$ as well.  On $[0,h_1]$, $V(h) = B_0h$ and hence
\begin{equation*}
(\mathcal{L}-r)V + mB_0h = (m-\delta)B_0 h \geq 0.
\end{equation*}
On $[h_2,h_3]$, $V(h) = B_0 + \alpha(h-1)$ and
\begin{equation*}
(\mathcal{L}-r)V + mB_0 = mB_0 - \alpha\delta h - r(B_0-\alpha).
\end{equation*}
At $h_3$ this is
\begin{equation*}
B_0\left(m-r\right) + r\alpha - \delta\frac{p_2}{1+p_2}\left(mB_0\left(\rcpdiff{r}{m}\right) + \alpha\right)  =\frac{1}{p_1}\left(B_0(m-r) + r\alpha\right) > 0,
\end{equation*}
where we used \eqref{E:p1_p2_ident}.  Thus, Proposition \ref{P:verification} applies, finishing the result.

\end{proof}


\begin{proof}[Proof of Proposition \ref{P:APRM_P3}]
We again drop all ``P'' subscripts, assume $m>r$, recall $p_1,p_2$ from \eqref{E:p1_p2} and heavily use \eqref{E:p1_p2_ident}. Value matching and smooth pasting at $h_3$ give
\begin{equation}\label{E:big_m_h3}
\begin{split}
h_3 &= h_3(\alpha) = \frac{p_2}{1+p_2}\left(\frac{mB_0}{\alpha}\left(\rcpdiff{r}{m}\right) + 1\right) = \frac{mB_0}{\alpha\delta}\left(\frac{p_1}{p_1-1} - \delmf\right)  + \frac{B_0/\alpha + p_2}{1+p_2}\\
\check{C}_2 &= -\frac{\alpha}{p_2}h_3^{1+p_2} < 0.
\end{split}
\end{equation}
Here, we obtained the second equality for $h_3$ using \eqref{E:p1_p2_ident}, and it shows $h_3 > 1$ because $\alpha < B_0$ and $p_1/(p_1-1) = m^*/\delta \leq m/\delta$.  Next, value matching at $h_2,1$ and $h_1$  give respectively
\begin{equation}\label{E:big_m_h2_comps}
\begin{split}
\wt{C}_1 &= \frac{1+p_2}{p_1+p_2}\alpha h_2^{1-p_1} - \frac{p_2}{p_1+p_2}\alpha h_2^{-p_1}\left(\frac{mB_0}{\alpha}\left(\rcpdiff{r}{m}\right)+1\right);\\
\wt{C}_2 &= \frac{p_1-1}{p_1+p_2}\alpha h_2^{1+p_2} -\frac{p_1}{p_1+p_2}\alpha h_2^{p_2} \left(\frac{mB_0}{\alpha}\left(\rcpdiff{r}{m}\right)+1\right)
\end{split}
\end{equation}
\begin{equation}\label{E:big_m_1_comps}
\begin{split}
\wt{C}_1 &= C_1 + \frac{1+p_2}{p_1(p_1+p_2)}\frac{mB_0}{\delta};\\
\wt{C}_2 &= C_2 - \frac{p_1-1}{p_2(p_1+p_2)}\frac{mB_0}{\delta}.
\end{split}
\end{equation}
\begin{equation}\label{E:big_m_h1_comps}
\begin{split}
C_1 &= -\frac{1+p_2}{p_1+p_2}h_1^{1-p_1}\frac{mB_0}{\delta}\left(1-\delmf\right)  < 0;\\
C_2 &= -\frac{p_1-1}{p_1+p_2}h_1^{1+p_2}\frac{mB_0}{\delta}\left(1-\delmf\right) < 0,
\end{split}
\end{equation}
Note, these are the same equations as \eqref{E:mid_m_h1_AB}, \eqref{E:mid_m_h2_AB} and \eqref{E:mid_m_h2_AB2} respectively. Therefore, solving for $h_1$ yields
\begin{equation}\label{E:big_m_h1_eqn2}
h_1 = \left(\frac{p_1\left(1-\delmf\right)}{1+\frac{\alpha\delta}{mB_0}p_1 h_2^{-p_1}(h_3-h_2)}\right)^{\frac{1}{p_1-1}}.
\end{equation}
Since $\delta < m$, if $h_2 \leq h_3$ then $h_1 > 0$.  Furthermore,
\begin{equation*}
h_1 < 1 \Longleftrightarrow \frac{mB_0}{\alpha\delta}\left(\frac{p_1-1}{p_1}-\delmf\right) < h_2^{-p_1}(h_3-h_2).
\end{equation*}
By assumption, $p_1\delta \leq (p_1-1)m$ so the left side above is non-negative. As for the right side, the map $h\to h^{-p_1}(h_3-h)$ is non-increasing for $0 < h \leq p_1h_3/(p_1-1)$ and at $h=1$ we have using \eqref{E:p1_p2_ident}
\begin{equation*}
\frac{mB_0}{\alpha\delta}\left(\frac{p_1-1}{p_1}-\delmf\right) < h_3 - 1 \quad \Longleftrightarrow  0 < \frac{1}{p_2}\left(\frac{B_0}{\alpha} - 1\right),
\end{equation*}
which holds as $0 < \alpha < B_0$.  Thus, there is a unique $h_0$ such that $h_1$ from \eqref{E:big_m_h1_eqn2} lies in $(0,1)$ provided $h_2$ lies in $(1,h_0)$, and
\begin{equation}\label{E:h0_def}
h_0 = \min\bra{ h > 1 \such h^{-p_1}(h_3-h) = \frac{mB_0}{\alpha\delta}\left(\frac{p_1-1}{p_1}-\delmf\right) } \in (1,h_3].
\end{equation}
where $h^* = h_3$ if and only if $p_1\delta = (p_1-1)m$.  Using $h_1$ from \eqref{E:big_m_h1_eqn2}, we obtain the same equation as \eqref{E:mid_m_h2_eqn2} for $h_2$,
and our goal is to show there is a unique solution $h_2$ lying in $(1,h^*)$ provided $0 < \alpha < B_0$. As before, dividing by $p_2$ yields the function
\begin{equation}\label{E:big_m_chi_def}
\begin{split}
\chi(h) &\dfn h^{p_2}\left(h -\frac{p_1(1+p_2)}{(p_1-1)p_2}h_3\right) + \frac{mB_0}{\alpha\delta}\left(\left(1-\delmf\right)\left(\frac{p_1\left(1-\delmf\right)}{1+\frac{\alpha\delta}{mB_0}p_1 h^{-p_1}(h_3-h)}\right)^{\frac{1+p_2}{p_1-1}} +\frac{1}{p_2}\right)
\end{split}
\end{equation}
The same argument as in Proposition \ref{P:APRM_P2} shows $\chi(1) > 0$ since there-in $(p_1-1)m < p_1\delta$ was only used to show $\alpha < \alpha^*$ implied $\alpha < B_0$, which we are now directly assuming (c.f. \eqref{E:mid_m_alpha_leq_B0}). Similarly, it still remains true that $\dot{\chi} < 0$ if and only if \eqref{E:mid_m_chi_decr_cond} holds, which we repeat here
\begin{equation*}
h^{p_1} + \frac{\alpha\delta}{mB_0} p_1\left(h_3- h\right) - p_1\left(1-\delmf\right) h > 0.
\end{equation*}
But, for $h\in (1,h_0]$ we have by definition
\begin{equation*}
h^{-p_1}(h_3-h) \geq \frac{mB_0}{\alpha\delta}\left(\frac{p_1-1}{p_1}-\delmf\right).
\end{equation*}
But this  implies
\begin{equation*}
\frac{\alpha\delta}{mB_0}p_1 (h_3-h) \geq p_1h^{p_1}\left(\frac{p_1-1}{p_1}-\delmf\right),
\end{equation*}
and hence
\begin{equation*}
h^{p_1} + \frac{\alpha\delta}{mB_0}p_1 (h_3-h) - p_1\left(1-\delmf\right)h \geq p_1h\left(1-\delmf\right)\left(h^{p_1-1}-1\right) > 0,
\end{equation*}
since $h>1$.  Thus, $\dot{\chi} < 0$ on $(1,h_0]$. Lastly, we wish to show that $\alpha < B_0$ implies $\chi(h_0) < 0$. To do so, we use the implicit function theorem. Indeed, with an eye towards \eqref{E:h0_def} define the function
\begin{equation*}
\varphi(\alpha,h) \dfn \frac{\alpha\delta}{mB_0} h^{-p_1}(h_3(\alpha)-h) - \left(\frac{p_1}{p_1-1}-\delmf\right).
\end{equation*}
We have already shown $\partial_h \varphi(\alpha,h) < 0$ for $0 < h <  p_1 h_3(\alpha)/(p_1-1)$, and plugging in for $h_3(\alpha)$ from \eqref{E:big_m_h3} and using \eqref{E:p1_p2_ident} we obtain
\begin{equation*}
\varphi(\alpha,h) = \left(\frac{p_1-1}{p_1} - \frac{p_2}{1+p_2}\frac{\delta}{m} - \frac{\alpha\delta}{mB_0}\left(h-\frac{p_2}{1+p_2}\right)\right)h^{-p_1} - \left(\frac{p_1-1}{p_1}-\delmf\right),
\end{equation*}
which is decreasing in $\alpha$ for $ p_2/(1+p_2) < h$.  Therefore, by the implicit function theorem we know
\begin{equation*}
0 = \varphi(\alpha,h_0(\alpha)) \Longrightarrow \partial_{\alpha} h_0(\alpha) = -\frac{\partial_\alpha \varphi}{\partial_h \varphi}\bigg|_{(\alpha,h_0(\alpha))} < 0,
\end{equation*}
so that $h_0(\alpha)$ is strictly decreasing in $\alpha$.  As we have already shown $\dot{\chi}(h_0) < 0$ it thus follows that $\alpha \to \chi(h_0(\alpha))$ is increasing.  This implies for $\alpha < B_0$ that $\chi(h_0(\alpha)) < \lim_{\wt{\alpha}\uparrow B_0} \chi(h_0(\wt{\alpha}))$.  From \eqref{E:big_m_h3} we see that $h_3(B_0) = p_1 m / ((p_1-1)\delta)$, $h_0(\alpha=B_0) = 1$ and hence $\chi(1) = 0$.  Therefore, $\chi(h_0) < 0$ for $\alpha < B_0$ which is what we wanted to show, because this ensures a unique $h_2 \in (1,h_0) \subset (1,h_3)$ such that $\chi(h_2) = 0$ and the associated $h_1 \in (0,1)$.

As for the constants, we know $\check{C}_2, C_1, C_2 < 0$. From \eqref{E:big_m_1_comps} we obtain $\wt{C}_2 < 0$.  As for $\wt{C}_1$, \eqref{E:big_m_h2_comps} implies
\begin{equation*}
\wt{C}_1 = -\frac{1+p_2}{p_1+p_2} \alpha h_2^{-p_1}\left(h_3 - h_2\right) < 0.
\end{equation*}
Therefore, repeating the steps in the proof of Proposition \ref{P:APRM_P2} exactly as in the line below \eqref{E:mid_m_wtC1} starting with ``This ensures...'' gives the result.
\end{proof}

\end{document}